\documentclass{article}%
\usepackage{amsfonts}
\usepackage{amsmath}
\usepackage{geometry}
\usepackage{amssymb}
\usepackage{graphicx}%
\providecommand{\U}[1]{\protect\rule{.1in}{.1in}}
\newtheorem{theorem}{Theorem}

\newtheorem{remark}[theorem]{Remark}

\newenvironment{proof}[1][Proof]{\noindent\textbf{#1.} }{\ \rule{0.5em}{0.5em}}
\begin{document}

\title{\textbf{General properties of the expansion methods of Lie algebras}}
\author{L. Andrianopoli$^{1,2}$, N. Merino$^{3}$, F. Nadal$^{1,4}$
and M. Trigiante$^{1,2}$\\$^{1}${\small Dipartimento di Scienza Applicata e Tecnologia (DISAT),
Politecnico di Torino, }\\{\small Corso Duca degli Abruzzi, 24, I-10129 Torino, Italia}\\$^{2}${\small Instituto Nazionale di Fisica Nucleare (INFN) Sezione di Torino,}\\{\small Via Pietro Giuria, 1  10125 Torino, Italia}\\$^{3}${\small Instituto de F\'{\i}sica, Pontificia Universidad Cat\'{o}lica de
Valpara\'{\i}so (PUCV), }\\{\small Av. Brasil 2950, Valpara\'{\i}so, Chile}\\$^{4}${\small Instituto de F\'{\i}sica Corpuscular (IFIC), Edificio Institutos
de Investigaci\'{o}n. }\\{\small c/ Catedr\'{a}tico Jos\'{e} Beltr\'{a}n, 2. E-46980 Paterna.
Espa\~{n}a.}}
\maketitle

\begin{abstract}

The study of the relation between Lie algebras and groups, and especially the derivation of new algebras from them, is a problem of great interest in mathematics and physics, because finding a new Lie group from an already known one also means that a new physical theory can be obtained from a known one. One of the procedures that allow to do so is called expansion of Lie algebras, and has been recently used in different physical applications - particularly in gauge theories of gravity. Here we report on further developments of this method, required to understand in a deeper way their consequences in physical theories. We have found theorems related to the preservation of some properties of the algebras under expansions that can be used as criteria and, more specifically, as necessary conditions to know if two arbitrary Lie algebras can be related by the some expansion mechanism. Formal aspects, such as the Cartan decomposition of the expanded algebras, are also discussed. Finally, an instructive example that allows to check explicitly all our theoretical results is also provided.
\end{abstract}

\section{Introduction}

Global and local symmetries of a physical system play an essential role in
modern theoretical physics and its physical applications. Apart from a
well-known relation between symmetries and conserved charges via the
Noether theorem, the knowledge of a symmetry group of a certain theory is
deeply built-in in its theoretical description and may lead to restrictive
`no-go' theorems. Introduction of new algebras is, in that sense, required by
possible solutions to these theoretical problems in the hope that they could
circumvent the `no-go' theorems. A beautiful example is the introduction of Lie
superalgebras that unify in a nontrivial way spacetime and internal symmetries
of the microscopic world, not allowed in a purely bosonic context by the
Coleman-Mandula theorem. In this way, the study of the relation between Lie algebras and groups, and especially the derivation of new algebras from them, is a problem of great
interest in mathematics and physics, because finding a new Lie group from an
already known one also means that a new physical theory can be obtained from a
known one. This is particularly useful, for example, in gauge theories (like
Yang-Mills and Chern-Simons theories) which have the symmetry group as a
fundamental ingredient.

Thus, setting aside the trivial problem of finding whether a Lie algebra is a
subalgebra of another one, there are, essentially, three different ways of
relating and/or obtaining new algebras from given ones. In fact, it was during
the second half of the XX century that certain mechanisms were developed to
obtain \textit{non-trivial}\footnote{By \textit{non-trivial relations} we mean
that these mechanisms allow us to obtain some Lie algebras starting from other
algebras that have completely different properties. Also, the original algebra
is not necessarily (though it could be in specific cases) contained as a
subalgebra of the algebra obtained by these processes.} relations between
different Lie groups and algebras. These mechanisms are known as \textit{contractions}
\cite{Segal,IW,WW,Saletan}, \textit{deformations} \cite{def1,def2,def4} and \textit{extensions}
\cite{AzLibro}, which all share the property of maintaining the dimension of the
original group or algebra. As we are going to see now, this work is focused on a generalization of the contraction procedure called \textit{expansion} that,
starting from a given algebra, permits us to generate algebras of a higher
dimension than the original one.

Expansions of Lie algebras are generalizations of the Weimar-Woods (WW)
contraction method \cite{WW} and were introduced some years ago in refs.
\cite{hs,aipv1,aipv2,aipv3}. While in a contraction a suitable
rescaling of some generators of the Lie algebra is done, in the expansion
method the starting point is to consider an algebra $\mathcal{G}$, with the
basis of generators $X_{i}$, as described by the Maurer-Cartan (MC) forms
$\omega^{i}(\phi)$ on the manifold $\mathcal{M}(\phi)$ of its associated group $G$. As it is known, the local
structure of the Lie group is encoded in the so called Maurer-Cartan
equations:
\[
d\omega^{k}=-\frac{1}{2}\,c_{ij}^{k}\,\omega^{i}\wedge\omega^{j}\,,
\]
which are just an equivalent description to the one given in terms of Lie
brackets among Lie algebra generators, $\left[  X_{i},X_{j}\right]  =c_{ij}^{k}\,X_{k}$. When a rescaling by
a parameter $\lambda$ is performed on some of the group coordinates $\phi^{i}%
$, the forms $\omega^{i}\left(  \phi,\lambda\right)  $ can be expanded as
power series in $\lambda$. Inserting these expansions back in the original MC
equations for $\mathcal{G}$, one obtains the MC equations of a new
finite-dimensional expanded Lie algebra. As shown in \cite{aipv1,aipv2,aipv3}
this method reproduces the In\"{o}n\"{u}-Wigner contractions \cite{IW} and its
generalizations in the sense of Weimar-Woods \cite{WW}. In these cases the
dimension of the algebra is preserved but, in general, this expansion method
also permits to obtain higher dimensional Lie algebras.

A generalization of the above method is the $S$-expansion\ \cite{irs} that
combines the structure constants of the algebra $\mathcal{G}$ with the inner
law of an Abelian semigroup $S$ in order to define the Lie bracket of a new
$S$-expanded algebra. When certain conditions are met, the method permits the
extraction of smaller algebras which are called \textit{resonant subalgebras}
and \textit{reduced algebras}. This method reproduces the results of the first
expansion method described above (that uses a parameter $\lambda$) for a
particular choice of the semigroup: one of the semigroups denoted as
$S_{E}^{\left(  n\right)  }$\ and its definition is given in \cite{irs}.
Since it reproduces the method above, it also reproduces all the generalized
In\"{o}n\"{u}-Wigner (IW) contractions, as is shown in Ref. \cite{irs}
explicitly\footnote{However, it is not clear if $S$-expansions exhaust all
possible contractions since there exist contractions that are not realized by
generalized IW-contractions \cite{Nesterenko1,Nesterenko6,Nesterenko7}.
Besides, there exist $S$-expansions which are not equivalent to contractions
as it was recently shown in \cite{Nesterenko}.}. The present article is
dedicated to a further development of the $S$-expansion method and its
application to the theory of Lie algebras.

Interesting physical applications of these methods, particularly of the $S$-expansion, appeared recently in the literature. One of the advantages, used in those applications, is that if we know the invariant tensors of a certain Lie algebra then the mechanism gives the invariant tensors for the expanded algebras, even if the last ones are not semisimple. This feature is especially useful in the construction of transgressions and Chern-Simons (CS) gauge theories of (super)gravity \cite{cham,zan,olea,olea2,olea3,olea4,mora,irs1} where the invariant tensors of the symmetry group under which the theory is invariant are fundamental ingredients of the theory. On one side, the problem of finding all the invariant tensors for a non-semisimple Lie algebra remains as an open problem until know. But this is not only an important mathematical problem but also one with physical relevance, because given an algebra the choice of the invariant tensor fixes the classes of interactions that may occur between the different fields in the theory. In fact, the standard procedure to obtain an invariant tensor of range $r$ is to use the symmetrized (super)trace for the product of $r$ generators in some matrix representation of the algebra. However this procedure has some limitations in the case of non-semisimple Lie algebras. A concrete example is provided in ref. \cite{irs1} where an eleven-dimensional gauge theory for the M Algebra (the maximal supersymmetric extension of the Poincar\'{e} algebra in $11$-dimensions) is constructed. It is shown that the supertrace is not a good choice to construct the theory, because it has so many vanish components that a transgression or a CS Lagrangian constructed from it depends only on the spin connection and it represents therefore a kind of \textit{exotic gravity}. Therefore with the supertrace as an invariant tensor it is not possible to reproduce general relativity neither to include fermionic fields or fields associated to central charges. However, by using the $S$-expansion procedure it is possible to consider the M algebra as an expansion of the $\mathfrak{osp}(32/1)$ algebra providing an invariant tensor, different from the supertrace, which leads to a theory with a richer structure. It is worthy to note here that while the $S$-expansion method does not solve the problem of classifying all invariant tensors for non semisimple algebras, at least it gives invariant tensors different from the supertrace that are useful for the construction of the mentioned gauge theories of gravity.

A dual formulation of the S-expansion
procedure for Lie algebra was also constructed in \cite{irs2}, which permits
to understand this procedure at the level of the Lagrangians. Another
interesting application of the expansion methods consists in establishing a
relation between General Relativity and CS gravity \cite{edelstein,K15}. As shown in \cite{K15},
General Relativity in five-dimensional spacetime may emerge at a special
critical point from a CS action. To achieve this result, both the Lie
algebra (called the $\mathfrak{B}$ algebra) and the symmetric invariant tensor that defines the CS Lagrangian are
constructed by means of the Lie algebra S-expansion method with the semigroup
$S_{E}^{\left(  3\right)  }$. It is also interesting to mention that a black hole \cite{quinz} and a cosmological solution \cite{gomez} have also been constructed for the CS theory constructed in terms of this expanded $\mathfrak{B}$ algebra.

On the other hand, in Refs. \cite{K11,K12} the $S$-expansion method was extended to the case of \textit{higher-order Lie
algebras} \cite{K10}. These kinds of algebras, defined in terms of a multilinear
antisymmetric product, are a particular case of the so-called \textit{strongly
homotopy Lie algebras} \cite{Spin3}, which in turn are related to the
structure of higher spin interactions \cite{Spin1,Spin2}. In ref.
\cite{K14}, the $S$-expansion method was also extended to study the expansion
of loop algebras, which are infinite dimensional and can be useful in
applications related to Higher Spin (HS) theories as, for example, in the
Vasiliev formulation thereof \cite{V0,V1,V2,V3,V4}. Finally, a procedure that permit us to contruct Casimir operators for the expanded algebra was recently found in \cite{Diaz}.

Therefore, as many physical applications appearing in the literature depend on
the possibility of relating two given algebras by some contraction or expansion
procedure, an exhaustive study of the general properties of this method is
needed to deeply understand their consequences in physical theories. In fact,
a feature that almost all those physical applications have in common is that
they depend on the following question: \textit{given two symmetry algebras,
can they be related by means of some contraction or expansion procedure?} Here we attempt to give the first steps to answer the above question\footnote{It
is worth remarking here that we use the $S$-expansion method because it
reproduces the other expansion methods for a particular choice of the
semigroup (the semigroups $S_{E}^{\left(  n\right)  }$\ mentioned before).}.
By studying the preserved properties of the algebras under expansions we find certain theorems that can be used as criteria and, more specifically, as necessary conditions to solve this problem. Then the relation might still exist or do not. The second part of the process, which consists in finding the relation explicitly with help of computer programs, is not presented here but mainly developed in \cite{Bianchi,Nadal}. In order to check our theoretical results,
examples are provided by performing all the possible expansions that can be made
for $\mathfrak{sl}\left(  2,\mathbb{R}\right)  $ with semigroups up to and
including order $6$.

This paper is organized as follows: In section \ref{desc}, we present a brief
technical description of the basic ingredients that we are going to use. In
section \ref{preserved}, we study under which conditions properties like solvability, nilpotency,
semisimplicity and compactness are preseved under $S$-expansions on each level of the procedure. Particularly, in section \ref{Cartandec}, we review
some aspects of the Cartan decomposition of the expanded algebra when
compactness is preserved. In section \ref{sl2} we make an exhaustive study of
the possible expansions that can be made for the algebra $\mathfrak{sl}\left(
2,\mathbb{R}\right)  $ and check our theoretical results. In section
\ref{sclassification} it is shown that a classification of semigroups in terms
of the eigenvalues of certain matrices, that are related to intrinsic
properties of the semigroup, is a useful tool to predict semisimplicity and
compactness properties for the expanded algebras. Section \ref{resum} is a
summary of the results and finally, in section \ref{Com}, some comments about
future applications are given.

\section{Preliminars}

\label{desc}

\subsection{The $S$-Expansion Procedure}

\label{sexpp}

In this section we briefly describe the general abelian semigroup expansion
procedure ($S$-expansion for short). We refer the interested reader to
ref.~\cite{irs} for further details.

Consider a Lie algebra $\mathcal{G}$ and a finite abelian semigroup
$S=\left\{  \lambda_{\alpha}\right\}  $. According to Theorem~3.1 from
ref.~\cite{irs}, the direct product
\begin{equation}
\mathcal{G}_{S}\text{\ }\mathcal{=}\text{\ }S\otimes\mathcal{G} \label{z1}%
\end{equation}
is also a Lie algebra. The elements of this expanded algebra are denoted by%
\begin{equation}
X_{\left(  i,\alpha\right)  }=\lambda_{\alpha}\otimes X_{i} \label{z2}%
\end{equation}
where the product is understood as a direct product of the matricial
representations of the generators $X_{i}$ of $\mathcal{G}$\ and the elements
$\lambda_{a}$ of the semigroup\footnote{The structure of the semigroup is
encoded in a quantity $K_{\alpha\beta}^{\gamma}$ which defines the composition
law: $\lambda_{\alpha}\cdot\lambda_{\beta}=K_{\alpha\beta}^{\gamma}%
\lambda_{\gamma}$.} $S$. The Lie product in $\mathcal{G}_{S}$ is defined as%

\begin{equation}
\left[  T_{\left(  i,\alpha\right)  },T_{\left(  j,\beta\right)  }\right]
=\lambda_{\alpha}\cdot\lambda_{\beta}\otimes\left[  T_{i},T_{j}\right]
=K_{\alpha\beta}^{\gamma}C_{ij}^{k}\left(  \lambda_{\lambda}\otimes
T_{k}\right)  =C_{\left(  i,\alpha\right)  \left(  j,\beta\right)  }^{\left(
k,\gamma\right)  }T_{\left(  k,\gamma\right)  } \label{z3}%
\end{equation}

\begin{equation}
C_{\left(  i,\alpha\right)  \left(  j,\beta\right)  }^{\left(  k,\gamma
\right)  }=K_{\alpha\beta}^{\gamma}C_{ij}^{k}%
\end{equation}
where $C_{ij}^{k}$ are the structure constants of $\mathcal{G}$ and
$K_{\alpha\beta}^{\gamma}$, which stores information
about the multiplication law of the semigroup, is called 2-selector. The set (\ref{z1}) with the
composition law (\ref{z3}) is called a $S$-expanded Lie algebra.

Interestingly, there are cases where it is possible to systematically extract
subalgebras from $S\otimes\mathcal{G}$ like, for example, if we start by
decomposing $\mathcal{G}$ in a direct sum of subspaces, as in $\mathcal{G}%
=\bigoplus_{p\in I}V_{p}$, where $I$ is a set of indices. The internal
structure of $\mathcal{G}$ can be codified through the mapping\footnote{Here
$2^{I}$ stands for the set of all subsets of $I$.} $i:I\otimes I\rightarrow
2^{I}$ according to%
\begin{equation}
\left[  V_{p},V_{q}\right]  \subset%
{\displaystyle\bigoplus\limits_{r\in i\left(  p,q\right)  }}
V_{r} \label{CDres1}%
\end{equation}
When the semigroup $S$ can be decomposed in subsets $S_{p}$, $S=\bigcup_{p\in
I}S_{p}$, such that they satisfy the \textit{resonant condition}\footnote{Here
$S_{p}\cdot S_{q}$ denotes the set of all the products of all elements from
$S_{p}$ with all elements from $S_{q}$.}%
\begin{equation}
S_{p}\cdot S_{q}\subset%
{\displaystyle\bigcap\limits_{r\in i\left(  p,q\right)  }}
S_{r} \label{CDres2}%
\end{equation}
\ then we have that
\begin{equation}
\mathcal{G}_{S,R}=\bigoplus_{p\in I}S_{p}\otimes V_{p} \label{gen_sub}%
\end{equation}
is a subalgebra of $\mathcal{G}_{S}$ which is called \textit{resonant
subalgebra} (see Theorem~4.2 from ref.~\cite{irs}). The commutation relations
of this subalgebra are given by:%
\begin{align}
\left[  T_{\left(  i_{p},\alpha_{p}\right)  },T_{\left(  j_{q},\beta
_{q}\right)  }\right]   &  =C_{\left(  i_{p},\alpha_{p}\right)  \left(
j_{q},\beta_{q}\right)  }^{\left(  k_{r},\gamma_{r}\right)  }T_{\left(
k_{r},\gamma_{r}\right)  }\label{lres1}\\
&  =K_{\alpha_{p}\beta_{q}}^{\gamma_{r}}C_{i_{p}j_{q}}^{k_{r}}T_{\left(
k_{r},\gamma_{r}\right)  }\nonumber
\end{align}
where a\ sum on the indexes $r\in I$ and $k_{r},\gamma_{r}$ on the
respective subspaces $V_{r}$, $S_{r}$ it is assumed. The indexes $\alpha
_{p}$, $\beta_{q}$ are fixed for each $p,q\in I$. In this work we write
$\mathcal{G}_{S,R}$ to denote this resonant subalgebra of the expanded algebra
$\mathcal{G}_{S}$.

An even smaller algebra can be obtained when there is a zero element in the
semigroup, i.e., an element $0_{S}\in S$ such that, for all $\lambda_{\alpha
}\in S$, $0_{S} \cdot\lambda_{\alpha}=0_{S}$. When this is the case, the whole
$0_{S}\otimes\mathcal{G}$ sector can be removed from the resonant subalgebra
by imposing $0_{S}\otimes\mathcal{G}=0$ (see Definition 3.3 from
ref.~\cite{irs}). The resulting algebra\ continues to be a Lie algebra and
here it will be denoted by $\mathcal{G}_{S,R}^{\text{red}}$.

\subsection{History of finite semigroups programs}

\label{history}

The number of finite non-isomorphic semigroups of order $n$ are given in the
following table:%

\begin{equation}%
\begin{tabular}
[c]{|l|l|l}\cline{1-2}%
order & $Q=\ $\# semigroups & \\\cline{1-2}%
1 & 1 & \\\cline{1-2}%
2 & 4 & \\\cline{1-2}%
3 & 18 & \\\cline{1-2}%
4 & 126 & [Forsythe '54]\\\cline{1-2}%
5 & 1,160 & [Motzkin, Selfridge '55]\\\cline{1-2}%
6 & 15,973 & [Plemmons '66]\\\cline{1-2}%
7 & 836,021 & [Jurgensen, Wick '76]\\\cline{1-2}%
8 & 1,843,120,128 & [Satoh, Yama, Tokizawa '94]\\\cline{1-2}%
9 & \textbf{52,989,400,714,478} & [Distler, Kelsey, Mitchell '09]\\\cline{1-2}%
\end{tabular}
\ \ \ \ \ \label{hist}%
\end{equation}
All the semigroups of order 4 have been classified by Forsythe in Ref.
\cite{n4}, of order 5 by Motzkin and Selfridge in Ref. \cite{n5}, of order 6
by Plemmons in Ref. \cite{n6-1,n6-2,n6-3}, of order 7 by J\"{u}rgensen and
Wick in Ref. \cite{n7}, and of order 8 by Satoh, Yama and Tokizawa in Ref.
\cite{n8}, and monoids and semigroups of order 9 by Distler and Kelsey in Ref.
\cite{n9-1,n9-2} and by Distler and Mitchell in Ref. \cite{n9-3}. Also, for
semigroups of order 9 the result can be found in Ref. \cite{n9-4}.

As shown in the table the problem of enumerating all the non-isomorphic finite
semigroups of a certain order is a non-trivial task. In fact, the number
$Q$ of semigroups increases very quickly with the order of the semigroup. 

In ref. \cite{Plemmons} a set of algorithms was given that allow us to make
certain calculations with finite semigroups. The first program, \textit{gen.f}%
, gives all the non-isomorphic semigroups of order $n$ for $n=1,2,...,8$.
\footnote{The order $n=9$ is non trivial and the algorithms in the mentionated
reference fail. As mentioned in the table, this non trivial problem was solved
in 2009 by Andreas Distler, Tom Kelsey \& James Mitchell. However in this
paper we are going to consider calculations with semigroups of at most order
$6$.} The input is the order, $n$, of the semigroups we want to obtain and the
\textit{output} is a list of all the non-isomorphic semigroups that exist at
this order. In this work, the elements of the semigroup are labeled by
$\lambda_{\alpha}$ with $\alpha=1,...,n$ and each semigroup will be denoted by
$S_{\left(  n\right)  }^{a}$ where the supra-index $a=1,...,Q$ identifies the
specific semigroup of order $n$.

Note also that this history includes the classification of all abelian semigroups (which are the ones used in the $S$-expansion mechanism introduced in last section) as a particular case. In fact, the second program in \cite{Plemmons}, \textit{com.f}, takes as \textit{input}
one of the mentioned lists for a certain order, picks up just the symmetric
tables and generates another list with all the abelian semigroups. For example
for $n=2$ the elements are labeled by $\left\{  \lambda_{1},\lambda
_{2}\right\}  $ and the program \textit{com.f} gives the following list of semigroups:%

\begin{equation}%
\begin{tabular}
[c]{l|ll}%
$S_{\left(  2\right)  }^{1}$ & $\lambda_{1}$ & $\lambda_{2}$\\\hline
$\lambda_{1}$ & $\lambda_{1}$ & $\lambda_{1}$\\
$\lambda_{2}$ & $\lambda_{1}$ & $\lambda_{1}$%
\end{tabular}
\ \ \ \text{, }%
\begin{tabular}
[c]{l|ll}%
$S_{\left(  2\right)  }^{2}$ & $\lambda_{1}$ & $\lambda_{2}$\\\hline
$\lambda_{1}$ & $\lambda_{1}$ & $\lambda_{1}$\\
$\lambda_{2}$ & $\lambda_{1}$ & $\lambda_{2}$%
\end{tabular}
\ \ \ \text{, }%
\begin{tabular}
[c]{l|ll}%
$S_{\left(  2\right)  }^{4}$ & $\lambda_{1}$ & $\lambda_{2}$\\\hline
$\lambda_{1}$ & $\lambda_{1}$ & $\lambda_{2}$\\
$\lambda_{2}$ & $\lambda_{2}$ & $\lambda_{1}$%
\end{tabular}
\ \ \label{list1}%
\end{equation}
Note that the semigroup $S_{\left(2\right)  }^{3}$ is not given in the list (\ref{list1}) because it is not abelian.

So in general the program \textit{com.f} of \cite{Plemmons} gives a list of
tables of all the abelian non-isomorphic semigroups of a certain
order\footnote{As the number of non-isomorphic semigroups increases very
quickly with the order $n$ (see table \ref{hist}), the mentioned lists are
very large for higher orders.} (up to order $8$). 

In section \ref{sl2}\ we are going to use those lists to perform all the
possible expansions of the algebra $\mathfrak{sl}\left(  2\right)  $ that can
be made (with semigroups up to the order $6$) and check explicitly the
theoretical properties we are going to find in the next section.

\section{Properties preserved under the $S$-expansion procedure}

\label{preserved}

\subsection{Expansion of solvable and nilpotent Lie algebras}

\label{solv}

It is known that a solvable algebra $\mathcal{G}$ is one for which the
sequence%
\begin{equation}
\mathcal{G}^{\left(  0\right)  }=\mathcal{G},\ \ \mathcal{G}^{\left(
1\right)  }=\left[  \mathcal{G}^{\left(  0\right)  },\mathcal{G}^{\left(
0\right)  }\right]  ,\ \ ...,\ \ \mathcal{G}^{\left(  n\right)  }=\left[
\mathcal{G}^{\left(  n-1\right)  },\mathcal{G}^{\left(  n-1\right)  }\right]
\ \label{sol1}%
\end{equation}
terminates, i.e., such that $\mathcal{G}^{\left(  n\right)  }=0$ for some $n$.
On the other hand, a nilpotent algebra $\mathcal{G}$ is one for which the
sequence%
\[
\mathcal{G}_{\left(  0\right)  }=\mathcal{G},\ \ \mathcal{G}_{\left(
1\right)  }=\left[  \mathcal{G}_{\left(  0\right)  },\mathcal{G}\right]
,\ \ ...,\ \ \mathcal{G}_{\left(  n\right)  }=\left[  \mathcal{G}_{\left(
n-1\right)  },\mathcal{G}\right]
\]
terminates, i.e., for which $\mathcal{G}_{\left(  n\right)  }=0$ for some $n$ (see
\cite{L2} for details).

In order to study the expansions of solvable and nilpotent Lie algebras, it
will be useful to have an expression of the solvability and nilpotency
condition in terms of the structure constants. We make this explicitly for the
solvable case.

Let $\left\{  X_{i}\right\}  $ be a basis of an algebra $\mathcal{G}$. Since
$\mathcal{G=G}^{\left(  0\right)  }$\ is solvable one can find, from the
commutation relations%
\begin{equation}
\left[  X_{i},X_{j}\right]  =C_{ij}^{k}X_{k} \label{sol2}%
\end{equation}
with $i,j=1,...\dim\mathcal{G}$, that there exist at least one value of $k$
such that $C_{ij}^{k}=0$. Let $k^{\left(  1\right)  }$ represent the set of
values such that $C_{ij}^{k^{\left(  1\right)  }}\neq0$. Then $k^{\left(
1\right)  }$ runs over all values of the basis elements $\left\{
X_{k}\right\}  $ except for those values for which
\begin{equation}
C_{ij}^{k\neq k^{\left(  1\right)  }}=0 \label{sol3}%
\end{equation}
It is clear then that the set $\left\{  X_{k^{\left(  1\right)  }}\right\}  $
is smaller than $\left\{  X_{k}\right\}  $ and $\left\{  X_{k^{\left(
1\right)  }}\right\}  \subset\left\{  X_{k}\right\}  $. Let us consider now
$\mathcal{G}^{\left(  2\right)  }=\left[  \mathcal{G}^{\left(  1\right)
},\mathcal{G}^{\left(  1\right)  }\right]  $ where $\mathcal{G}^{\left(
1\right)  }=\left[  \mathcal{G}^{\left(  0\right)  },\mathcal{G}^{\left(
0\right)  }\right]  =\left\{  X_{k^{\left(  1\right)  }}\right\}  $ according
to the notation established above. We have%
\begin{equation}
\left[  X_{i^{\left(  1\right)  }},X_{j^{\left(  1\right)  }}\right]
=C_{i^{\left(  1\right)  }j^{\left(  1\right)  }}^{k^{\left(  2\right)  }%
}X_{k^{\left(  2\right)  }} \label{sol4}%
\end{equation}
and again as $\mathcal{G}$ is solvable the index $k^{\left(  2\right)  }$ must
run in a smaller subset with respect to that of $k^{\left(  1\right)  }$.
Therefore, if the algebra is solvable, there must exists some $n$ for which
$\mathcal{G}^{\left(  n\right)  }=0$ and $\mathcal{G}^{\left(  n-1\right)  }$
is abelian, i.e.,
\begin{equation}
\left[  X_{i^{\left(  n-1\right)  }},X_{j^{\left(  n-1\right)  }}\right]
=C_{i^{\left(  n-1\right)  }j^{\left(  n-1\right)  }}^{k^{\left(  n\right)  }%
}X_{k^{\left(  n\right)  }}=0\text{.} \label{sol5}%
\end{equation}

Then, for solvable algebras there exist some $n$ such that the range of values of $k^{(n)}$ is empty. Equivalently, the solvability of an algebra can be expressed in terms of its
structure constants as the condition that there is some $n$ for which
$C_{i^{\left(  n-1\right)  }j^{\left(  n-1\right)  }}^{k^{\left(  n\right)  }%
}=0$. In a similar way it can be shown that an algebra is nilpotent if the
structure constants $C_{i^{\left(  n-1\right)  }j}^{k^{\left(  n\right)  }}$
vanish for some $n$.

\begin{theorem}
Let $\left\{  X_{i}\right\}  $ be a basis of a solvable Lie algebra
$\mathcal{G}$, $S=\left\{  \lambda_{\alpha}\right\}  $ a finite abelian
semigroup and
\begin{equation}
\mathcal{G}_{S}=S\otimes\mathcal{G}=\left\{  \lambda_{\alpha}\otimes
X_{i}\right\}  =\left\{  X_{\left(  i,\alpha\right)  }\right\}  \label{sol6}%
\end{equation}
the $S$-expanded algebra, which satisfies%
\begin{equation}
\left[  X_{\left(  i,\alpha\right)  },X_{\left(  j,\beta\right)  }\right]
=C_{\left(  i,\alpha\right)  \left(  j,\beta\right)  }^{\left(  k,\gamma
\right)  }X_{\left(  k,\gamma\right)  }=C_{ij}^{k}K_{\alpha\beta}^{\gamma
}X_{\left(  k,\gamma\right)  }\text{.}\label{sol7}%
\end{equation}
Then the expanded algebra $\mathcal{G}_{S}=S\otimes\mathcal{G}$ is solvable.
\end{theorem}

\begin{proof}
Let's consider the following sequence for the expanded algebra%
\begin{equation}
\mathcal{G}_{S}^{\left(  0\right)  }=\mathcal{G}_{S},\ \ \mathcal{G}%
_{S}^{\left(  1\right)  }=\left[  \mathcal{G}_{S}^{\left(  0\right)
},\mathcal{G}_{S}^{\left(  0\right)  }\right]  ,\ \ ...,\ \ \mathcal{G}%
_{S}^{\left(  n\right)  }=\left[  \mathcal{G}_{S}^{\left(  n-1\right)
},\mathcal{G}_{S}^{\left(  n-1\right)  }\right]  \ \text{.} \label{sol8}%
\end{equation}
For $\mathcal{G}_{S}^{\left(  n\right)  }$ we have%
\begin{equation}
\left[  X_{\left(  i^{\left(  n-1\right)  },\alpha^{\left(  n-1\right)
}\right)  },X_{\left(  j^{\left(  n-1\right)  },\beta^{\left(  n-1\right)
}\right)  }\right]  =C_{i^{\left(  n-1\right)  }j^{\left(  n-1\right)  }%
}^{k^{\left(  n\right)  }}K_{\alpha^{\left(  n-1\right)  }\beta^{\left(
n-1\right)  }}^{\gamma^{\left(  n\right)  }}X_{\left(  k^{\left(  n\right)
},\gamma^{\left(  n\right)  }\right)  }\text{.} \label{sol9}%
\end{equation}
So as $\mathcal{G}$ is solvable by hypothesis, then there exists some $n$ for
which the sequence (\ref{sol8}) terminates, i.e., for which $\mathcal{G}%
_{S}^{\left(  n\right)  }=0$ (or in terms of the structure constants, for
which $C_{\left(  i^{\left(  n-1\right)  },\alpha^{\left(  n-1\right)
}\right)  \left(  j^{\left(  n-1\right)  },\beta^{\left(  n-1\right)
}\right)  }^{\left(  k^{\left(  n\right)  },\gamma^{\left(  n\right)
}\right)  }=C_{i^{\left(  n-1\right)  }j^{\left(  n-1\right)  }}^{k^{\left(
n\right)  }}K_{\alpha^{\left(  n-1\right)  }\beta^{\left(  n-1\right)  }%
}^{\gamma^{\left(  n\right)  }}=0$).
\end{proof}

The result that the expansion of a solvable Lie algebra is solvable too was
also found, in a different way, in Ref. \cite{Nesterenko} but only at the
first level of the expanded algebra, $\mathcal{G}_{S}=\mathcal{G}\otimes S$.
Here we prove that this result holds also for the resonant and the reduced algebra.

\begin{theorem}
The resonant subalgebra $\mathcal{G}_{S,R}$ (defined in (\ref{gen_sub})) of
the expanded algebra $\mathcal{G}_{S}$ is always solvable if the original
algebra $\mathcal{G}$ is solvable.
\end{theorem}

\begin{proof}
By the theory of classification of Lie algebras, any subalgebra of a solvable
algebra must be solvable.

\end{proof}

\begin{theorem}
Consider the expansion of a solvable Lie algebra $\mathcal{G}$ with a zero
element. Then the reduced algebra $\mathcal{G}_{S}^{\text{red}}$ is always solvable.
\end{theorem}

\begin{proof}
According with the $S$-expansion procedure (of ref. \cite{irs}) when the
semigroup has a zero element $0_{S}$, $S=\left\{  \lambda_{\alpha}%
,0_{S}\right\}  $ the commutation relations of the expanded algebra
$\mathcal{G}_{S}$ are given by%
\begin{align*}
\left[  X_{\left(  i,\alpha\right)  },X_{\left(  j,\beta\right)  }\right]   &
=C_{ij}^{k}K_{\alpha\beta}^{\gamma}X_{\left(  k,\gamma\right)  }+C_{ij}%
^{k}K_{\alpha\beta}^{0}X_{\left(  k,0\right)  }\\
\left[  X_{\left(  i,0\right)  },X_{\left(  j,\beta\right)  }\right]   &
=C_{ij}^{k}X_{\left(  k,0\right)  }\\
\left[  X_{\left(  i,0\right)  },X_{\left(  j,0\right)  }\right]   &
=C_{ij}^{k}X_{\left(  k,0\right)  }%
\end{align*}
and the $0_{S}$-reduced algebra $\mathcal{G}_{S}^{\text{red}}$ is given by%
\[
\left[  X_{\left(  i,\alpha\right)  },X_{\left(  j,\beta\right)  }\right]
=C_{ij}^{k}K_{\alpha\beta}^{\gamma}X_{\left(  k,\gamma\right)  }%
\]
So if $\mathcal{G}$ is solvable then there exists some $n$ for which
$C_{i^{\left(  n-1\right)  }j^{\left(  n-1\right)  }}^{k^{\left(  n\right)  }%
}=0$. Therefore, for the same $n$%
\[
C_{\left(  i^{\left(  n-1\right)  },\alpha^{\left(  n-1\right)  }\right)
\left(  j^{\left(  n-1\right)  },\beta^{\left(  n-1\right)  }\right)
}^{\left(  k^{\left(  n\right)  },\gamma^{\left(  n\right)  }\right)
}=C_{i^{\left(  n-1\right)  }j^{\left(  n-1\right)  }}^{k^{\left(  n\right)
}}K_{\alpha^{\left(  n-1\right)  }\beta^{\left(  n-1\right)  }}^{\gamma
^{\left(  n\right)  }}=0
\]
and we can conclude that if $\mathcal{G}$ is solvable then the reduced algebra
$\mathcal{G}_{S}^{\text{red}}$ is solvable, too.
\end{proof}

In a similar way it can be directly shown that if $\mathcal{G}$ is nilpotent,
then $\mathcal{G}_{S}$, $\mathcal{G}_{S,R}$ and $\mathcal{G}_{S}^{\text{red}}$
are nilpotent too.

\subsection{Expansion of semisimple and compact Lie algebras}

\label{semis}

An algebra $\mathcal{G}$ is \textit{semi-simple} if its Killing-Cartan metric,
defined in terms of the structure constants by%
\[
g_{ij}=C_{ik}^{l}C_{jl}^{k}\text{,}%
\]
is non degenerate, i.e., if $\det\left(  g_{ij}\right)  \neq0$. On the other
hand, $g_{ij}$ is diagonalizable so if we denote by $\left(  \mu_{i}\right)  $
its spectra of eigenvalues then a semi-simple Lie algebra $\mathcal{G}$ is
\textit{compact} if and only if $\mu_{i}<0$\ (see \cite{L2} for details).

If we now perform the expansion of a semi-simple algebra $\mathcal{G}$, the
Killing-Cartan metric of the expanded algebra $\mathcal{G}_{S}$ will be given
by:%
\begin{align}
g_{\left(  i,\alpha\right)  \left(  j,\beta\right)  }  &  =C_{\left(
i,\alpha\right)  \left(  k,\gamma\right)  }^{\left(  l,\lambda\right)
}C_{\left(  j,\beta\right)  \left(  l,\lambda\right)  }^{\left(
k,\gamma\right)  }\label{mario}\\
&  =K_{\alpha\gamma}^{\lambda}K_{\beta\lambda}^{\gamma}C_{ik}^{l}C_{jl}%
^{k}\nonumber
\end{align}
Let us define the following matrices:%
\begin{equation}
\mathbf{g}^{E}\equiv\left(  g_{\left(  i,\alpha\right)  \left(  j,\beta
\right)  }\right)  \ ;\ \ \mathbf{g}^{S}\equiv\left(  g_{\alpha\beta}%
^{S}\right)  =\left(  K_{\alpha\gamma}^{\lambda}K_{\beta\lambda}^{\gamma
}\right)  \ ;\ \ \mathbf{g}\equiv\left(  g_{ij}\right)  =\left(  C_{ik}%
^{l}C_{jl}^{k}\right)  \ . \label{mario2}%
\end{equation}
From eq. (\ref{mario}) it follows that the first matrix is the Kronecker
product of the last two:%
\begin{equation}
\mathbf{g}^{E}=\mathbf{g}^{S}\otimes\mathbf{g\ .} \label{mario3}%
\end{equation}
Both $\mathbf{g}^{S}$ and $\mathbf{g}$ are diagonalizable. Let us denote by
$\left(  \xi_{\alpha}\right)  $ and $\left(  \mu_{i}\right)  $ their spectra
of eigenvalues respectively. From the general theory of Kronecker products we
know that:

\begin{itemize}
\item The eigenvalues of $\mathbf{g}^{E}$\ are $\left(  \xi_{\alpha}\mu
_{i}\right)  $;

\item $\det\left(  \mathbf{g}^{E}\right)  =\det\left(  \mathbf{g}^{S}\right)
^{\dim\left(  \mathcal{G}\right)  }\det\left(  \mathbf{g}\right)  ^{\left\vert
S\right\vert }$ where $\left\vert S\right\vert $ is the order of the semigroup.
\end{itemize}

Thus if $\mathcal{G}$ is semisimple ($\det\left(  \mathbf{g}\right)  \neq0$),
$\mathcal{G}_{S}$ is semisimple if and only if $\det\left(  \mathbf{g}%
^{S}\right)  \neq0$. If $\mathcal{G}$ is compact ($\mu_{i}<0$), $\mathcal{G}%
_{S}$ is compact only if $\xi_{\alpha}>0$. Thus the real form of
$\mathcal{G}_{S}$ strongly depends on the signs of $\xi_{\alpha}$.

When considering a decomposition $\mathcal{G}=\bigoplus_{p\in I}V_{p}\,$and
$S=\bigcup_{p\in I}S_{p}$, where $I$ is a set of indices, the Killing-Cartan metric
of the expanded algebra is given by:%

\begin{align}
\mathbf{g}^{E}  &  =\mathbf{g}^{S}\otimes\mathbf{g}\label{lres5}\\
&  =\left(
\begin{array}
[c]{ccc}%
g_{\alpha_{0}\beta_{0}}^{S} & g_{\alpha_{0}\beta_{1}}^{S} & \cdots\\
g_{\alpha_{1}\beta_{0}}^{S} & g_{\alpha_{1}\beta_{1}}^{S} & \\
\vdots &  & \ddots
\end{array}
\right)  \otimes\left(
\begin{array}
[c]{ccc}%
g_{i_{0}j_{0}} & g_{i_{0}j_{1}} & \cdots\\
g_{i_{1}j_{0}} & g_{i_{1}j_{1}} & \\
\vdots &  & \ddots
\end{array}
\right) \nonumber
\end{align}
where the sets $S_{p}$ and subspaces $V_{q}$ are mixed all together. If
$\mathcal{G}$ and $S$ have a structure given by (\ref{CDres1}-\ref{CDres2}),
then the Killing-Cartan metric of the resonant subalgebra $\mathcal{G}_{S,R}$
of eq. (\ref{gen_sub}) is given by%
\begin{align}
g_{\left(  i_{p},\alpha_{p}\right)  \left(  j_{q},\beta_{q}\right)  }^{E,R}
&  =C_{\left(  i_{p},\alpha_{p}\right)  \left(  k_{r},\gamma_{r}\right)
}^{\left(  l_{s},\rho_{s}\right)  }C_{\left(  j_{q},\beta_{q}\right)  \left(
l_{s},\rho_{s}\right)  }^{\left(  k_{r},\gamma_{r}\right)  }
\label{genKKf_res}\\
&  =K_{\alpha_{p}\gamma_{r}}^{\rho_{s}}K_{\beta_{q}\rho_{s}}^{\gamma_{r}%
}C_{i_{p}k_{r}}^{l_{s}}C_{j_{q}l_{s}}^{k_{r}}\nonumber
\end{align}
and from its matrix form,%

\begin{align}
\mathbf{g}^{E,R}  &  =\left(  g_{\left(  i_{p},\alpha_{p}\right)  \left(
j_{q},\beta_{q}\right)  }^{E,R}\right) \label{lres6}\\
&  =\left(
\begin{array}
[c]{ccc}%
g_{\left(  i_{0},\alpha_{0}\right)  \left(  j_{0},\beta_{0}\right)  }^{E,R} &
g_{\left(  i_{0},\alpha_{0}\right)  \left(  j_{1},\beta_{1}\right)  }^{E,R} &
\cdots\\
g_{\left(  i_{1},\alpha_{1}\right)  \left(  j_{0},\beta_{0}\right)  }^{E,R} &
g_{\left(  i_{1},\alpha_{1}\right)  \left(  j_{1},\beta_{1}\right)  }^{E,R} &
\\
\vdots &  & \ddots
\end{array}
\right)  \text{,}\nonumber
\end{align}
it can be seen that sets $S_{p}$ and subspaces $V_{q}$ are mixed in a
\textit{resonant} form.

From the representation theory of Lie algebras (see Appendix A) it turns out
that if $\mathcal{G}_{S}$ and $\mathcal{G}_{S,R}$ are semisimple then
$\mathcal{G}_{S}$ is completely reducible with respect to the
action of $\mathcal{G}_{S,R}\ $and therefore%
\begin{equation}
\left[  \mathcal{G}_{S,R},\frac{\mathcal{G}_{S}}{\mathcal{G}_{S,R}}\right]
\subset\frac{\mathcal{G}_{S}}{\mathcal{G}_{S,R}}\text{.} \label{teor}%
\end{equation}
Then it is possible to obtain the Killing-Cartan metric for $\mathcal{G}%
_{S,R}$ in terms of the one for $\mathcal{G}_{S}$
\begin{equation}
g_{\left(  i_{p},\alpha_{p}\right)  \left(  j_{q},\beta_{q}\right)  }%
^{E,R}=\alpha g_{\left(  i_{p},\alpha_{p}\right)  \left(  j_{q},\beta
_{q}\right)  }^{E} \label{sup}%
\end{equation}
where
\begin{equation}
\alpha=\left(  1+\frac{d_{\mathcal{G}_{S}/\mathcal{G}_{S,R}}c_{\mathcal{G}%
_{S}/\mathcal{G}_{S,R}}}{d_{\mathcal{G}_{S,R}}}\right)  ^{-1} \label{alfaa}%
\end{equation}
is a positive number, defined in the Appendix A, characterizing $\mathcal{G}%
_{S}/\mathcal{G}_{S,R}$ and $\mathcal{G}_{S,R}$. On the other hand, if
$\mathcal{G}_{S}$ and $\mathcal{G}_{S,R}$ are not semisimple, then the general
expresion for the Killing-Cartan form will be given just by equation
(\ref{genKKf_res}).

An interesting case, which often appears in physical applications, is to take
$I=\left\{  0,1\right\}  $, i.e., $\mathcal{G}=V_{0}\oplus V_{1}$ with
$\mathcal{G}$ having a subspace structure given by\footnote{As an example we
will perform, in section \ref{sl2}, all the possible expansions of an algebra
that satisfy this particular kind of decomposition.}%
\begin{align}
\left[  V_{0},V_{0}\right]   &  \subset V_{0}\label{lres8}\\
\left[  V_{0},V_{1}\right]   &  \subset V_{1}\nonumber\\
\left[  V_{1},V_{1}\right]   &  \subset V_{0}\text{.}\nonumber
\end{align}
If the semigroup has a resonant decomposition $S=S_{0}\cup S_{1}$ satisfying%
\begin{align}
S_{0}\times S_{0} &  \subset S_{0}\label{lres9}\\
S_{0}\times S_{1} &  \subset S_{1}\nonumber\\
S_{1}\times S_{1} &  \subset S_{0}\text{,}\nonumber
\end{align}
with $S_{0}=\left\{  \lambda_{\alpha_{0}}\right\}  $ and $S_{1}=\left\{
\lambda_{\alpha_{1}}\right\}  $, then the resonant subalgebra (\ref{gen_sub})
takes the form%
\begin{equation}
\mathcal{G}_{S,R}=\left(  S_{0}\otimes V_{0}\right)  \oplus\left(
S_{1}\otimes V_{1}\right)  \label{lres9_2}%
\end{equation}
with a Killing-Cartan metric given by%
\begin{equation}
\mathbf{g}^{E,R}=\left(  g_{\left(  i_{p},\alpha_{p}\right)  \left(
j_{q},\beta_{q}\right)  }^{E,R}\right)  =\left(
\begin{array}
[c]{cc}%
g_{\left(  i_{0},\alpha_{0}\right)  \left(  j_{0},\beta_{0}\right)  }^{E,R} &
g_{\left(  i_{0},\alpha_{0}\right)  \left(  j_{1},\beta_{1}\right)  }^{E,R}\\
g_{\left(  i_{1},\alpha_{1}\right)  \left(  j_{0},\beta_{0}\right)  }^{E,R} &
g_{\left(  i_{1},\alpha_{1}\right)  \left(  j_{1},\beta_{1}\right)  }^{E,R}%
\end{array}
\right)  \nonumber
\end{equation}
and where%
\begin{align*}
g_{\left(  i_{0},\alpha_{0}\right)  \left(  j_{0},\beta_{0}\right)  }^{E,R} &
=K_{\alpha_{0}\gamma_{0}}^{\rho_{0}}K_{\beta_{0}\rho_{0}}^{\gamma_{0}}%
C_{i_{0}k_{0}}^{l_{0}}C_{j_{0}l_{0}}^{k_{0}}+K_{\alpha_{0}\gamma_{1}}%
^{\rho_{1}}K_{\beta_{0}\rho_{1}}^{\gamma_{1}}C_{i_{0}k_{1}}^{l_{1}}%
C_{j_{0}l_{1}}^{k_{1}}\\
g_{\left(  i_{0},\alpha_{0}\right)  \left(  j_{1},\beta_{1}\right)  }^{E,R} &
=0\ ;\ g_{\left(  i_{1},\alpha_{1}\right)  \left(  j_{0},\beta_{0}\right)
}^{E,R}=0\\
g_{\left(  i_{1},\alpha_{1}\right)  \left(  j_{1},\beta_{1}\right)  }^{E,R} &
=K_{\alpha_{1}\gamma_{0}}^{\rho_{1}}K_{\beta_{1}\rho_{1}}^{\gamma_{0}}%
C_{i_{1}k_{0}}^{l_{1}}C_{j_{1}l_{1}}^{k_{0}}+K_{\alpha_{1}\gamma_{1}}%
^{\rho_{0}}K_{\beta_{1}\rho_{0}}^{\gamma_{1}}C_{i_{1}k_{1}}^{l_{0}}%
C_{j_{1}l_{0}}^{k_{1}}~\text{.}%
\end{align*}
Moreover $\mathcal{G}_{S}$ and $\mathcal{G}_{S,R}$ are semisimple then we
can use (\ref{sup}) to write%
\begin{align*}
g_{\left(  i_{0},\alpha_{0}\right)  \left(  j_{0},\beta_{0}\right)  }^{E,R} &
=\alpha g_{\left(  i_{0},\alpha_{0}\right)  \left(  j_{0},\beta_{0}\right)
}^{E}=\alpha g_{\alpha_{0}\beta_{0}}^{S}g_{i_{0}j_{0}}\\
g_{\left(  i_{1},\alpha_{1}\right)  \left(  j_{1},\beta_{1}\right)  }^{E,R} &
=\alpha g_{\left(  i_{1},\alpha_{1}\right)  \left(  j_{1},\beta_{1}\right)
}^{E}=\alpha g_{\alpha_{1}\beta_{1}}^{S}g_{i_{1}j_{1}}%
\end{align*}
or equivalently%
\begin{equation}
\mathbf{g}^{E,R}=\alpha\left(
\begin{array}
[c]{cc}%
\mathbf{g}\left(  S_{0}\right)  \otimes\mathbf{g}\left(  V_{0}\right)   & 0\\
0 & \mathbf{g}\left(  S_{1}\right)  \otimes\mathbf{g}\left(  V_{1}\right)
\end{array}
\right)  \label{res_final}%
\end{equation}
where%
\begin{align}
\mathbf{g}^{S}\left(  S_{0}\right)   &  =\left(  g_{\alpha_{0}\beta_{0}}%
^{S}\right)  =\left(  K_{\alpha_{0}\gamma_{0}}^{\rho_{0}}K_{\beta_{0}\rho_{0}%
}^{\gamma_{0}}+K_{\alpha_{0}\gamma_{1}}^{\rho_{1}}K_{\beta_{0}\rho_{1}%
}^{\gamma_{1}}\right)  \label{zdef1}\\
\mathbf{g}^{S}\left(  S_{1}\right)   &  =\left(  g_{\alpha_{1}\beta_{1}}%
^{S}\right)  =\left(  K_{\alpha_{1}\gamma_{0}}^{\rho_{1}}K_{\beta_{1}\rho_{1}%
}^{\gamma_{0}}+K_{\alpha_{1}\gamma_{1}}^{\rho_{0}}K_{\beta_{1}\rho_{0}%
}^{\gamma_{1}}\right)  \label{zdef1b}%
\end{align}%
\begin{align}
\mathbf{g}\left(  V_{0}\right)   &  =\left(  g_{i_{0}j_{0}}\right)  =\left(
C_{i_{0}k_{0}}^{l_{0}}C_{j_{0}l_{0}}^{k_{0}}+C_{i_{0}k_{1}}^{l_{1}}%
C_{j_{0}l_{1}}^{k_{1}}\right)  \label{zdef2}\\
\mathbf{g}\left(  V_{1}\right)   &  =\left(  g_{i_{1}j_{1}}\right)  =\left(
C_{i_{1}k_{0}}^{l_{1}}C_{j_{1}l_{1}}^{k_{0}}+C_{i_{1}k_{1}}^{l_{0}}%
C_{j_{1}l_{0}}^{k_{1}}\right)  \label{zdef2b}%
\end{align}
From the fact that $\alpha$ is a positive number the eigenvalues of
$\mathbf{g}^{E,R}$, $\left(  \alpha\mu_{i_{0}}\xi_{\alpha_{0}},\alpha
\mu_{i_{1}}\xi_{\alpha_{1}}\right)  $, allow to analize the compactness of
$\mathcal{G}_{S,R}$ in terms of the eigenvalues of the matrices $\mathbf{g}%
^{S}\left(  S_{0}\right)  $ and $\mathbf{g}^{S}\left(  S_{1}\right)  $, i.e.,
in terms of $\xi_{\alpha_{0}}$ and $\xi_{\alpha_{1}}$.

Finally, if the semigroup has a zero element a reduction can be performed and
the result is similar\footnote{Here we consider the case when a reduction of
this specific resonant subalgebra is performed. Note however that the $0_{S}%
$-reduction procedure is a step that can be performed independently from the
extraction of the resonant subalgebra. Explicit examples of this fact are given
in section 4, figure 1.}:%
\begin{equation}
\mathbf{g}^{E,R,\text{red}}=\alpha\left(
\begin{array}
[c]{cc}%
\mathbf{g}^{\text{red}}\left(  S_{0}\right)  \otimes\mathbf{g}\left(
V_{0}\right)  & 0\\
0 & \mathbf{g}^{\text{red}}\left(  S_{1}\right)  \otimes\mathbf{g}\left(
V_{1}\right)
\end{array}
\right)  \label{res_red_final}%
\end{equation}
where $\mathbf{g}^{\text{red}}\left(  S_{0}\right)  $ and $\mathbf{g}%
^{\text{red}}\left(  S_{1}\right)  $ are calculated by using (\ref{zdef1}%
-\ref{zdef1b}) but considering that their indices do not take value on 
the zero element $0_{S}$ and considering that $K_{\alpha\beta}^{\gamma}=0$ if one of
the indices $\alpha,\beta,\gamma$ is evaluated on the zero element.

\subsection{Expansion of a general Lie algebra}

\label{gener}

By the Levi-Macev theorem (see \cite{L2} for details), an arbitrary Lie
algebra $\mathcal{G}$ always has a decomposition%
\begin{equation}
\mathcal{G=}N\uplus S \label{arb1}%
\end{equation}
where $N$ is the radical (maximal solvable ideal) and $S$ is semisimple.

In this section the semigroup used to perform the expansions will be denoted
by the symbol $\mathcal{S}$, because in this section we are using $S$ to
denote the semisimple part of the Levi-Malcev decomposition of the arbitrary
algebra (\ref{arb1}).

\bigskip

\textbf{The expanded algebra:}

Let's call $X_{i_{N}}$ and $X_{i_{S}}$ the generators of $N$ and $S$
respectively. And let $\mathcal{S}=\left\{  \lambda_{\alpha}\right\}  $ be a
finite abelian semigroup. Then the expanded algebra of $\mathcal{G}$
(\ref{arb1}) is given by%
\[
\mathcal{G}_{\mathcal{S}}=\left\{  X_{\left(  i_{N},\alpha\right)
},X_{\left(  i_{S},\beta\right)  }\right\}
\]
and we have the following commutation relations:%
\begin{align*}
\left[  X_{\left(  i_{N},\alpha\right)  },X_{\left(  j_{N},\beta\right)
}\right]   &  =C_{\left(  i_{N},\alpha\right)  \left(  j_{N},\beta\right)
}^{\left(  k_{N},\gamma\right)  }X_{\left(  k_{N},\gamma\right)  }\\
\left[  X_{\left(  i_{N},\alpha\right)  },X_{\left(  j_{S},\beta\right)
}\right]   &  =C_{\left(  i_{N},\alpha\right)  \left(  j_{S},\beta\right)
}^{\left(  k_{N},\gamma\right)  }X_{\left(  k_{N},\gamma\right)  }\\
\left[  X_{\left(  i_{S},\alpha\right)  },X_{\left(  j_{S},\beta\right)
}\right]   &  =C_{\left(  i_{S},\alpha\right)  \left(  j_{S},\beta\right)
}^{\left(  k_{S},\gamma\right)  }X_{\left(  k_{S},\gamma\right)  }%
\end{align*}
where we have used the fact tha $\mathcal{G}$ is a semi-direct sum,
(\ref{arb1}). So the expanded algebra has also a semi-direct structure:%
\[
\mathcal{G}_{\mathcal{S}}=E_{N}\uplus E_{S}%
\]
where $E_{N}=\left\{  X_{\left(  i_{N},\alpha\right)  }\right\}  $ is an ideal (which is also nilpotent being an expansion of a nilpotent algebra) and $E_{S}=\left\{  X_{\left(  i_{S},\alpha\right)  }\right\}  $ is a
subalgebra. The important issue to mention here is that the subalgebra
$E_{S}=\left\{  X_{\left(  i_{S},\alpha\right)  }\right\}  $ is not
necessarily semi-simple, because it is the expansion of a semi-simple algebra
$S$ which, as we found in the last section, is arbitrary, i.e., semisimplicity
can be broken. Thus it must have again a Levi-Malcev decomposition:
\[
E_{S}=N^{\prime}\uplus S_{\text{exp}}%
\]
where $N^{\prime}=\left\{  X_{\left(  i_{S},\alpha\right)  _{N^{\prime}}%
}\right\}  $ and $S_{\text{exp}}=\left\{  X_{\left(  i_{S},\alpha\right)
_{S}}\right\}  $ and
\[
\left(  i_{S},\alpha\right)  =\left\{  \left(  i_{S},\alpha\right)
_{N^{\prime}},\left(  i_{S},\alpha\right)  _{S}\right\}
\]
expresses the fact that the expansion of the semi-simple part is arbitrary and
will have a Levi-Malcev decomposition. Then, the expanded algebra is also
arbitrary and must have a Levi-Malcev decomposition:
\begin{align*}
\mathcal{G}_{\mathcal{S}}  &  =E_{N}\uplus\left(  N^{\prime}\uplus
S_{\text{exp}}\right) \\
&  =\left(  E_{N}\uplus N^{\prime}\right)  \uplus S_{\text{exp}}\\
&  =N_{\text{exp}}\uplus S_{\text{exp}}%
\end{align*}
where $N_{\text{exp}}=E_{N}\uplus N^{\prime}$ \ is the radical of the expanded
algebra $\mathcal{G}_{\mathcal{S}}$. In fact, the last statement can be easily
proved by showing that $N_{\text{exp}}$ is the maximal solvable ideal, which
is made in Appendix B. In this way we have the following:

\begin{theorem}
If \textit{$\mathcal{G}$} is an arbitrary Lie algebra with a Levi-Malcev
decomposition \textit{$\mathcal{G=}N\uplus S=\left\{  X_{i_{N}}\right\}
\uplus\left\{  X_{i_{S}}\right\}  $}, then the expanded algebra is also
arbitrary and has a Levi-Malcev decomposition given by%
\[
\mathcal{G}_{\mathcal{S}}\mathcal{=}N_{\text{exp}}\uplus S_{\text{exp}}%
\]
where \textit{$N_{\text{exp}}=E_{N}\uplus N^{\prime}$} is the radical of
\textit{$\mathcal{G}_{\mathcal{S}}$}, \textit{$E_{S}=N^{\prime}\uplus
S_{\text{exp}}$} is the expansion of \textit{$S$}, \textit{$N^{\prime}$} is
the radical of \textit{$E_{S}$}, \textit{$S_{\text{exp}}$} is the semisimple
part of \textit{$E_{S}$} and \textit{$\mathcal{G}_{\mathcal{S}}$} with%
\begin{align*}
E_{N}  &  =\left\{  X_{\left(  i_{N},\alpha\right)  }\right\} \\
E_{S}  &  =\left\{  X_{\left(  i_{S},\alpha\right)  }\right\} \\
N^{\prime}  &  =\left\{  X_{\left(  i_{S},\alpha\right)  _{N^{\prime}}%
}\right\} \\
S_{\text{exp}}  &  =\left\{  X_{\left(  i_{S},\alpha\right)  _{S}}\right\}
\end{align*}

\end{theorem}

The resonant subalgebra $\mathcal{G}_{\mathcal{S},R}$ of this abitrary Lie
algebra $\mathcal{G}_{\mathcal{S}}$ must also have a Levi-Malcev decomposition
and its specific form can be studied in each case with a similar procedure.

\bigskip

\textbf{The $0_{S}$-reduced algebra:}

Let%
\'{}%
s consider the case when the semigroup has a zero element, so $\mathcal{S}$
$=\left\{  \lambda_{\Gamma}\right\}  =\left\{  \lambda_{\alpha}\right\}
\cup\left\{  0_{\mathcal{S}}\right\}  $. The reduced algebra is given by%
\[
\left[  X_{\left(  i,\alpha\right)  },X_{\left(  j,\beta\right)  }\right]
=C_{ij}^{k}K_{\alpha\beta}^{\gamma}X_{\left(  k,\gamma\right)  }%
\]
where the index $\alpha,\beta...$ runs over all the non-zero elements of the semigroup.

In the last section we obtained that for the algebra $\mathcal{G=}N\uplus
S=\left\{  X_{i_{N}}\right\}  \uplus\left\{  X_{i_{S}}\right\}  $, the
complete expanded algebra is given by
\[
\mathcal{G}_{\mathcal{S}}=E_{N}\uplus N^{\prime}\uplus S_{\text{exp}}%
\]
where $E_{N}=\left\{  X_{\left(  i_{N},\Gamma\right)  }\right\}  $,
$N^{\prime}=\left\{  X_{\left(  i_{S},\Gamma\right)  _{N^{\prime}}}\right\}  $
and $S_{\text{exp}}=\left\{  X_{\left(  i_{S},\Gamma\right)  _{S}}\right\}  $.
Then the reduced algebra is then given by%
\[
\mathcal{G}_{\mathcal{S}}^{\text{red}}=\left\{  X_{\left(  i_{N}%
,\alpha\right)  },X_{\left(  i_{S},\alpha\right)  _{N^{\prime}}},X_{\left(
i_{S},\alpha\right)  _{S}}\right\}
\]
with commutation relations:%
\begin{align*}
\left[  X_{\left(  i_{N},\alpha\right)  },X_{\left(  j_{N},\beta\right)
}\right]   &  =C_{\left(  i_{N},\alpha\right)  \left(  j_{N},\beta\right)
}^{\left(  k_{N},\gamma\right)  }X_{\left(  k_{N},\gamma\right)  }\\
\left[  X_{\left(  i_{N},\alpha\right)  },X_{\left(  j_{S},\beta\right)
_{N^{\prime}}}\right]   &  =C_{\left(  i_{N},\alpha\right)  \left(
j_{S},\beta\right)  _{N^{\prime}}}^{\left(  k_{N},\gamma\right)  }X_{\left(
k_{N},\gamma\right)  }\\
\left[  X_{\left(  i_{N},\alpha\right)  },X_{\left(  j_{S},\beta\right)  _{S}%
}\right]   &  =C_{\left(  i_{N},\alpha\right)  \left(  j_{S},\beta\right)
_{S}}^{\left(  k_{N},\gamma\right)  }X_{\left(  k_{N},\gamma\right)  }\\
\left[  X_{\left(  i_{S},\alpha\right)  _{N^{\prime}}},X_{\left(  j_{S}%
,\beta\right)  _{N^{\prime}}}\right]   &  =C_{\left(  i_{S},\alpha\right)
_{N^{\prime}}\left(  j_{S},\beta\right)  _{N^{\prime}}}^{\left(  k_{S}%
,\gamma\right)  _{N^{\prime}}}X_{\left(  k_{S},\gamma\right)  _{N^{\prime}}}\\
\left[  X_{\left(  i_{S},\alpha\right)  _{N^{\prime}}},X_{\left(  j_{S}%
,\beta\right)  _{S}}\right]   &  =C_{\left(  i_{S},\alpha\right)  _{N^{\prime
}}\left(  j_{S},\beta\right)  _{S}}^{\left(  k_{S},\gamma\right)  _{N^{\prime
}}}X_{\left(  k_{S},\gamma\right)  _{N^{\prime}}}\\
\left[  X_{\left(  i_{S},\alpha\right)  _{S}},X_{\left(  j_{S},\beta\right)
_{S}}\right]   &  =C_{\left(  i_{S},\alpha\right)  _{S}\left(  j_{S}%
,\beta\right)  _{S}}^{\left(  k_{S},\gamma\right)  _{S}}X_{\left(
k_{S},\gamma\right)  _{S}}%
\end{align*}
so that
\[%
\begin{tabular}
[c]{l}%
$E_{N}^{\text{red}}=\left\{  X_{\left(  i_{N},\alpha\right)  }\right\}  $ is
an ideal of $\mathcal{G}_{\mathcal{S}}^{\text{red}}$\\
$N^{\prime\text{red}}=\left\{  X_{\left(  i_{S},\alpha\right)  _{N^{\prime}}%
}\right\}  $ is a subalgebra of $\mathcal{G}_{\mathcal{S}}^{\text{red}}$ and
an ideal of $E_{S}^{\prime}=N^{\prime\text{red}}\uplus S_{\text{exp}%
}^{\text{red}}$\\
$S_{\text{exp}}^{\text{red}}=\left\{  X_{\left(  i_{S},\alpha\right)  _{S}%
}\right\}  $ is a subalgebra of $\mathcal{G}_{\mathcal{S}}^{\text{red}}$,
\end{tabular}
\ \ \
\]
therefore, the semidirect structure of $\mathcal{G}_{\mathcal{S}}^{\text{red}%
}$ is confirmed,%
\[
\mathcal{G}_{\mathcal{S}}^{\text{red}}=E_{N}^{\text{red}}\uplus N^{\prime
\text{red}}\uplus S_{\text{exp}}^{\text{red}}%
\]
Now it only remains to prove that $S_{\text{exp}}^{\text{red}}$ is
semi-simple. We can write $S_{\text{exp}}=\left\{  X_{\left(  i_{S}%
,\Gamma\right)  _{S}}\right\}  =\left\{  X_{\left(  i_{S},\alpha\right)  _{S}%
},X_{\left(  i_{S},0_{S}\right)  _{S}}\right\}  $. Since $S_{\text{exp}%
}^{\left(  0\right)  }\equiv\left\{  X_{\left(  i_{S},0_{\mathcal{S}}\right)
_{S}}\right\}  $ is an ideal of the semi-simple Lie algebra $S_{\text{exp}}$,
by a Corollary of the \textit{Structure theorem} for semisimple Lie algebras,
$S_{\text{exp}}^{\left(  0\right)  }$ has to be semisimple itself. In other
words, if we express $S_{\text{exp}}$ as a direct sum of commuting simple Lie
algebras:%
\[
S_{\text{exp}}=%
{\displaystyle\bigoplus\limits_{k=0}^{l}}
S_{\text{exp},k}\ ;\ \ \left[  S_{\text{exp},k},S_{\text{exp},k^{\prime}%
}\right]  =0\text{ for }k\neq k^{\prime}\text{,}%
\]
then $S_{\text{exp}}^{\left(  0\right)  }$ coincides with the direct sum of a
subset of these simple algebras:%
\[
S_{\text{exp}}^{\left(  0\right)  }=%
{\displaystyle\bigoplus\limits_{k\in I}}
S_{\text{exp},k}%
\]
with $I$ a subset of $\left\{  0,1,...,l\right\}  $. Thus $S_{\text{exp}%
}^{\text{red}}=S_{\text{exp}}\ominus S_{\text{exp}}^{\left(  0\right)  }$ is
semisimple. Then we have that%
\[
\mathcal{G}_{\mathcal{S}}^{\text{red}}=N_{\text{exp}}^{\text{red}}\uplus
S_{\text{exp}}^{\text{red}}%
\]
is the Levi-Malcev decomposition for the reduced algebra where $N_{\text{exp}%
}^{\text{red}}=E_{N}^{\text{red}}\uplus N^{\prime\text{red}}$ is the radical
and $S_{\text{exp}}^{\text{red}}=S_{\text{exp}}\ominus S_{\text{exp}}^{\left(
0\right)  }$ the semi-simple part.

\subsection{The Cartan decomposition under the S-expansion}

\label{Cartandec}


As a brief explanation of what the Cartan decomposition is, let us cite some
theorems (theorem 6.3 and 7.1) and definitions from \cite{L4}:

\textbf{Theorem:} \textit{Every semi-simple Lie algebra over $%
\mathbb{C}
$ has a compact real form, $\mathcal{G}_{k}$.}

In fact if $\left\{  H_{\alpha},E_{\alpha},E_{-\alpha}\right\}  $ is the Cartan-Weyl
basis, then \ the compact real form $\mathcal{G}_{k}$ is given by
\begin{equation}
\mathcal{G}_{k}\mathcal{=}\sum%
\mathbb{R}
\left(  iH_{\alpha}\right)  +\sum%
\mathbb{R}
\left(  E_{\alpha}-E_{-\alpha}\right)  +\sum%
\mathbb{R}
\left(  i\left(  E_{\alpha}+E_{-\alpha}\right)  \right)  \label{BCD1}%
\end{equation}
The proof is given in ref. \cite{L4}.

\textbf{Theorem:} \textit{Let $\mathcal{G}_{0}$ be a real semisimple Lie
algebra, $\mathcal{G}$ its complex form and $\mathcal{U}$ any compact real
form of $\mathcal{G}$. Let $\sigma$ and $\tau$ be conjugations of
$\mathcal{G}$ with respect to $\mathcal{G}_{0}$ and $\mathcal{U}%
$\ respectively. Then there exists an automorphism $\varphi$ of $\mathcal{G}$
such that the compact real form $\varphi\left(  \mathcal{U}\right)  $ is
invariant under $\sigma$.}

So having these results, it is possible to give (see ref. \cite{L4}) the
definition of a Cartan decomposition of a semisimple algebra:

\textbf{Definition:} \textit{Let $\mathcal{G}_{0}$ be a semisimple real Lie
algebra, $\mathcal{G}$ its complexification and $\sigma$ a conjugation of
$\mathcal{G}$ with respect to $\mathcal{G}_{0}$. Then}
\begin{equation}
\mathcal{G}_{0}=\mathcal{T}_{0}+\mathcal{P}_{0}\text{,} \label{BCD2}%
\end{equation}
\textit{where $\mathcal{T}_{0}$ is a subalgebra, is called a Cartan
decomposition if there exists a compact real form, $\mathcal{G}_{k}$, of
$\mathcal{G}$ such that}
\begin{equation}
\sigma\left(  \mathcal{G}_{k}\right)  \subset\mathcal{G}_{k}\text{ \ and \
\begin{tabular}
[c]{l}%
$\mathcal{T}_{0}=\mathcal{G}_{0}\cap\mathcal{G}_{k}$\\
$\mathcal{P}_{0}=\mathcal{G}_{0}\cap\left(  i\mathcal{G}_{k}\right)  $%
\end{tabular}
} \label{BCD3}%
\end{equation}

The first two cited theorems imply that each real semi-simple Lie algebra
$\mathcal{G}_{0}$ has a Cartan decomposition. It can be demonstrated also that
$\mathcal{T}_{0}$ is the maximal compactly embedded subalgebra of
$\mathcal{G}_{0}$. From now on, we are going to use the symbol "$0$" to
characterize structures related to real and semisimple Lie
algebras\footnote{The only exception to this convention will be when we
consider an algebra $\mathcal{G}$ and a semigroup $S$ having a decomposition:
\begin{align*}
\mathcal{G}  &  =V_{0}+V_{1}\\
S  &  =S_{0}\cup S_{1}%
\end{align*}
}.

Note also that the compact real form $\mathcal{G}_{k}$ can be constructed in
terms of $\mathcal{T}_{0}$\ and $\mathcal{P}_{0}$ as:%
\begin{equation}
\mathcal{G}_{k}=\mathcal{T}_{0}+i\mathcal{P}_{0}\text{.} \label{mario_res0}%
\end{equation}
which satisfies%
\begin{align}
\left[  \mathcal{T}_{0},\mathcal{T}_{0}\right]   &  \subset\mathcal{T}%
_{0}\label{mario_res}\\
\left[  \mathcal{T}_{0},i\mathcal{P}_{0}\right]   &  \subset i\mathcal{P}%
_{0}\nonumber\\
\left[  i\mathcal{P}_{0},i\mathcal{P}_{0}\right]   &  \subset\mathcal{T}%
_{0}\nonumber
\end{align}

\bigskip

\textbf{The expanded semisimple algebra:}

Let $\mathcal{G}_{0}$ be a real semi-simple Lie algebra, $\mathcal{G}$ its
complexification and $\mathcal{G}_{k}$ its compact real form. If we perform an
expansion of the compact real form $\mathcal{G}_{k}$ with a semigroup that
preserves the compactness, then the expansion with this semigroup of
$\mathcal{G}_{0}$ will also preserve the semisimplicity. This is true because
the first condition implies that the eigenvalues $\xi_{\alpha}$ are all
positive, so that $\det\mathbf{g}^{S}\neq0$. Therefore it must have a Cartan
decomposition, whose explicit form is given by the following

\begin{theorem}
Let $\mathcal{G}_{0}=\mathcal{T}_{0}+\mathcal{P}_{0}$ be the Cartan
decomposition of the real semi-simple algebra $\mathcal{G}_{0}$, $\mathcal{G}$
its complexification and $\mathcal{G}_{k}$ its compact real form. Let
$\left\{  X_{i^{\left(  0\right)  }}\right\}  _{i^{\left(  0\right)  }%
=1}^{\dim\mathcal{G}_{0}}$, $\left\{  X_{i_{0}^{\left(  0\right)  }}\right\}
_{i_{0}^{\left(  0\right)  }=1}^{\dim\mathcal{T}_{0}}$, $\left\{
X_{i_{1}^{\left(  0\right)  }}\right\}  _{i_{1}^{\left(  0\right)  }=1}%
^{\dim\mathcal{P}_{0}}$, $\left\{  X_{i}\right\}  _{i=1}^{\dim\mathcal{G}}$
and $\left\{  X_{i^{\left(  k\right)  }}\right\}  _{i^{\left(  k\right)  }%
=1}^{\dim\mathcal{G}_{k}}$ be the bases of $\mathcal{G}_{0}$, $\mathcal{T}%
_{0}$, $\mathcal{P}_{0}$, $\mathcal{G}$ and $\mathcal{G}_{k}$ respectively. If
the semigroup used in the expansion preserves the compactness ($\xi_{\alpha}$
are all positive) of the compact real form $\mathcal{G}_{k}$, i.e, if
\begin{equation}
\mathcal{G}_{k,S}=S\otimes\mathcal{G}_{k}=\left\{  X_{\left(  i^{\left(
k\right)  },\alpha\right)  }\right\}  \label{CC_1}%
\end{equation}
is compact then%
\begin{equation}
\mathcal{G}_{0,S}=\mathcal{T}_{0,S}+\mathcal{P}_{0,S}\text{,}\label{SCD0}%
\end{equation}
where
\begin{align}
\mathcal{T}_{0,S} &  =\left\{  X_{\left(  i_{0}^{\left(  0\right)  }%
,\alpha\right)  }\right\}  \label{SCD_01}\\
\mathcal{P}_{0,S} &  =\left\{  X_{\left(  i_{1}^{\left(  0\right)  }%
,\alpha\right)  }\right\}  \label{SCD_02}%
\end{align}
is the Cartan decomposition of the expanded algebra $\mathcal{G}_{0,S}$.
\end{theorem}

\begin{proof}
See appendix C.
\end{proof}

\bigskip

\textbf{The semisimple resonant subalgebra}


The Cartan decomposition of a real semisimple Lie algebra%
\begin{equation}
\mathcal{G}_{0}=\mathcal{T}_{0}+\mathcal{P}_{0}\label{Cdres4}%
\end{equation}
has the subspace structure (\ref{lres8}) with $V_{0}=\mathcal{T}_{0}$ and
$V_{1}=\mathcal{P}_{0}$, i.e.,%
\begin{align}
\left[  \mathcal{T}_{0},\mathcal{T}_{0}\right]   &  \subset\mathcal{T}%
_{0}\label{CDres5}\\
\left[  \mathcal{T}_{0},\mathcal{P}_{0}\right]   &  \subset\mathcal{P}%
_{0}\nonumber\\
\left[  \mathcal{P}_{0},\mathcal{P}_{0}\right]   &  \subset\mathcal{T}%
_{0}\nonumber
\end{align}
which is a particular case of (\ref{CDres1}). So if we perform a special
expansion by choosing a semigroup having a decomposition $S=S_{0}\cup S_{1}$
satisfying the \textit{resonant condition} (\ref{lres9}), then the resonant
subalgebra of $\mathcal{G}_{0,S,R}$ is given by:%
\begin{align}
\mathcal{G}_{0,S,R} &  =\left(  S_{0}\otimes\mathcal{T}_{0}\right)  +\left(
S_{1}\otimes\mathcal{P}_{0}\right)  \label{CDreS1}\\
&  =\mathcal{T}_{0,S,R}+\mathcal{P}_{0,S,R}\nonumber
\end{align}
where%
\begin{align}
\mathcal{T}_{0,S,R} &  =\left(  S_{0}\otimes\mathcal{T}_{0}\right)  =\left\{
X_{\left(  i_{0}^{\left(  0\right)  },\alpha_{0}\right)  }\right\}
\label{CDreS2}\\
\mathcal{P}_{0,S,R} &  =\left(  S_{1}\otimes\mathcal{P}_{0}\right)  =\left\{
X_{\left(  i_{1}^{\left(  0\right)  },\alpha_{1}\right)  }\right\}
\label{CDreS3}%
\end{align}

Then we have the following

\begin{theorem}
If the semigroup $S$ used in the expansion is such that all the eigenvalues
$\xi_{\alpha}$ of the martrix $\mathbf{g}^{S}$ are positive then the resonant
subalgebra $\mathcal{G}_{0,S,R}$ (which is obviously semisimple because
$\det\mathbf{g}^{S}\neq0$) has a Cartan decomposition given precisely by
(\ref{CDreS1}-\ref{CDreS3}).
\end{theorem}

\begin{proof}
See Appendix C.
\end{proof}

These results can be presented in a more general form by relaxing the asumption
that $S$ preserves compactness, i.e., $\xi_{\alpha}>0$. In fact, one can
consider the broader class of semigroups preserving semisimplicity, i.e.,
$\det\mathbf{g}^{S}\neq0$. Nevertheless, under this condition the explicit
form of the Cartan decomposition is not necessarily given by eqs.
(\ref{CDreS1}-\ref{CDreS3}) (eqs. (\ref{SCD0}-\ref{SCD_02})). The reason is
that the expression for $\mathcal{G}_{k,S,R}$ given in eq. (\ref{f3}) (for
$\mathcal{G}_{k,S}$ given in eq. (\ref{CC_1})) is not necessarily compact:
compact generators may become non-compact and vice-versa depending on the
signature of $\mathbf{g}^{S}$. It is definitely true that to demand
$\det\mathbf{g}^{S}\neq0$ is enough for the expanded algebra to admit a Cartan
decomposition because under this condition $\mathcal{G}_{0,S,R}$
($\mathcal{G}_{0,S}$) is a real semisimple Lie algebra. To find it explicitly
one would have to diagonalize $\mathbf{g}^{E,R}$ ($\mathbf{g}^{E}$) to find
the new compact and non compact generators, which would not longer be
expressible as tensor products of elements of $S$ and $\mathcal{G}$, but rather
as combination of tensor products. For example, the diagonalization of
$\mathbf{g}^{E}$ is given by%
\[
\mathbf{g}^{E,D}=\tilde{O}^{T}\mathbf{g}^{E}\tilde{O}%
\]
where $\tilde{O}_{\ \ \left(  \alpha^{\prime},i^{\prime}\right)  }^{\left(
\alpha,i\right)  }=\mathcal{O}_{\ \alpha^{\prime}}^{\alpha}\hat{O}%
_{\ i^{\prime}}^{i}$ is an orthogonal matrix, $\mathcal{O}_{\ \alpha^{\prime}%
}^{\alpha}$ and $\hat{O}_{\ i^{\prime}}^{i}$ are the orthogonal matrices
diagonalizing $\mathbf{g}^{S}$ and $\mathbf{g}$ respectively. Then the new
basis of generators%
\[
X_{\left(  \alpha^{\prime},i^{\prime}\right)  }^{\prime}=X_{\left(
\alpha,i\right)  }\tilde{O}_{\ \ \left(  \alpha^{\prime},i^{\prime}\right)
}^{\left(  \alpha,i\right)  }%
\]
split into two subsets, corresponding to negative and positive eigenvalues,
generating the subspaces $\mathcal{T}_{0,S}$ and $\mathcal{P}_{0,S}$, so that
the Cartan decomposition will be given by $\mathcal{G}_{0,S}=\mathcal{T}%
_{0,S}+\mathcal{P}_{0S}$. A similar procedure can be done in the case of the
resonant subalgebra.

\bigskip

\textbf{The semisimple reduction of the resonant subalgebra}

The same results obtained above can be extended when the reduction of the resonant
subalgebra, $\mathcal{G}_{k,S,R}^{\text{red}}$, is compact. In this case we
have%
\begin{equation}
\mathcal{G}_{k,S,R}^{\text{red}}=\left(  S_{0}^{\text{red}}\otimes
\mathcal{T}_{0}\right)  +\left(  S_{1}^{\text{red}}\otimes i\mathcal{P}%
_{0}\right)  \label{f5}%
\end{equation}
where $S_{0}^{\text{red}}=S_{0}-\left\{  0_{S}\right\}  $ and $S_{1}%
^{\text{red}}=S_{1}-\left\{  0_{S}\right\}  $. The Cartan decomposition of
$\mathcal{G}_{0,S,R}^{\text{red}}$ is given by%
\begin{align}
\mathcal{G}_{0,S,R}^{\text{red}} &  =\left(  S_{0}^{\text{red}}\otimes
\mathcal{T}_{0}\right)  +\left(  S_{1}^{\text{red}}\otimes\mathcal{P}%
_{0}\right)  \label{f6}\\
&  =\mathcal{T}_{0,S,R}^{\text{red}}+\mathcal{P}_{0,S,R}^{\text{red}}\nonumber
\end{align}
with%
\begin{align}
\mathcal{T}_{0,S,R}^{\text{red}} &  =\mathcal{G}_{0,S,R}^{\text{red}}%
\cap\mathcal{G}_{k,S,R}^{\text{red}}=\left(  S_{0}^{\text{red}}\otimes
\mathcal{T}_{0}\right)  =\left\{  X_{\left(  i_{0}^{\left(  0\right)  }%
,\alpha_{0}^{\text{red}}\right)  }\right\}  \label{f7}\\
\mathcal{P}_{0,S,R}^{\text{red}} &  =\mathcal{G}_{0,S,R}^{\text{red}}%
\cap\left(  i\mathcal{G}_{k,S,R}^{\text{red}}\right)  =\left(  S_{1}%
^{\text{red}}\otimes\mathcal{P}_{0}\right)  =\left\{  X_{\left(
i_{1}^{\left(  0\right)  },\alpha_{1}^{\text{red}}\right)  }\right\}
\nonumber
\end{align}
where $\alpha_{0}^{\text{red}}$ and $\alpha_{1}^{\text{red}}$ are indices
running on $S_{0}^{\text{red}}$ and $S_{1}^{\text{red}}$ respectively.

\section{Expansions of $\mathfrak{sl}(2,\mathbb{R} )$, an instructive example}

\label{sl2}

Consider the algebra $\mathcal{G}_{0}=\mathfrak{sl}(2,\mathbb{R})$ and its
complex form $\mathcal{G=}\mathfrak{sl}(2,\mathbb{C})$. As known,
$\mathcal{G}_{k}=\mathfrak{su}(2)$ is the compact real form of $\mathfrak{sl}%
(2,\mathbb{C})$ so the Cartan decomposition of $\mathfrak{sl}(2,\mathbb{R})$,%
\begin{equation}
\mathcal{G}_{0}=\mathfrak{sl}(2,\mathbb{R})=\mathcal{T}_{0}+\mathcal{P}%
_{0}\text{,}\label{not00}%
\end{equation}%
\begin{align}
\mathcal{T}_{0} &  =\left\{  -i\sigma_{2}\right\}  \label{not01}\\
\mathcal{P}_{0} &  =\left\{  \sigma_{1},\sigma_{3}\right\}  \nonumber
\end{align}
is such that:%
\begin{align}
\mathcal{T}_{0} &  =\mathfrak{sl}(2,\mathbb{R})\cap\mathfrak{su}%
(2)\label{not02}\\
\mathcal{P}_{0} &  =\mathfrak{sl}(2,\mathbb{R})\cap i\left(  \mathfrak{su}%
(2)\right)  \nonumber
\end{align}
This can be checked directly by considering that the basis of $SU(2)\,$is
given by:%

\begin{equation}
\mathcal{G}_{k}=\mathfrak{su}(2)=\left\{  i\sigma_{2},i\sigma_{1},i\sigma
_{3}\right\}  =\mathcal{T}_{0}+i\mathcal{P}_{0}\label{not03}%
\end{equation}
where $\sigma_{1}$, $\sigma_{2}$ and $\sigma_{3}$ are the usual Pauli matrices.
It is also directly checked that $\mathfrak{sl}(2,\mathbb{R})$ has subspace
structure (\ref{lres8}) with $V_{0}=\mathcal{T}_{0}$ and $V_{1}=\mathcal{P}%
_{0}$ and that the Killing-Cartan metric for $\mathfrak{sl}(2,\mathbb{R})$ is given
by
\begin{align}
\mathbf{g} &  =\left(
\begin{array}
[c]{cc}%
\mathbf{g}\left(  V_{0}\right)   & 0\\
0 & \mathbf{g}\left(  V_{1}\right)
\end{array}
\right)  =diag\left(  -8,8,8\right)  \label{not05}\\
\mathbf{g}\left(  V_{0}\right)   &  =-8\ \ ;\ \ \mathbf{g}\left(
V_{1}\right)  =diag\left(  8,8\right)  \nonumber
\end{align}

According to our notation:
\begin{equation}
X_{1}=-i\sigma_{2}\text{, }X_{2}=\sigma_{1}\text{ and }X_{3}=\sigma
_{3},\label{not1}%
\end{equation}
so the elements of the expanded algebra will have the form
\begin{equation}
X_{\left(  i,\alpha\right)  }=\lambda_{\alpha}\otimes X_{i}\label{not2}%
\end{equation}
as in (\ref{z2}), with $i=1,2,3$ and $\alpha$ an index in a certain semigroup
$S$.

In what follows we study the expansions of $\mathfrak{sl}(2,\mathbb{R})$ with
the standard semigroups introduced in section \ref{history}. The results
obtained in section \ref{preserved} will also be checked with this example.

\subsection{Classification of the different kinds of expansions}

We are going to study the properties of all the possible expansions of
$\mathfrak{sl}(2,\mathbb{R})$ that can be made by using semigroups of order
$n=1,2,...,6$. These expansions can be classified as one of the following kinds:

\begin{itemize}
\item expansions with all the abelian semigroups

\item expansions with all the abelian semigroups with a zero element

\item expansions with all the abelian semigroups with a resonant decomposition
of the form (\ref{lres9})

\item expansions with all the abelian semigroups with a zero element and
simultaneously with a resonant decomposition of the form (\ref{lres9}%
)\footnote{It is interesting to note that a semigroup can have more than one
resonant decomposition leading then to different expanded algebras.}.
\end{itemize}

To perform all these possible expansions we use the following algorithm:

First we identify all the semigroups of a certain order $n$ (we have limited
this up to order $6$) satisfying the conditions enumerated above. For example,
for $n=3$ the results are given in figure \ref{fig:fig0}:

\begin{figure}[th]
\centering
\includegraphics[scale=0.75]{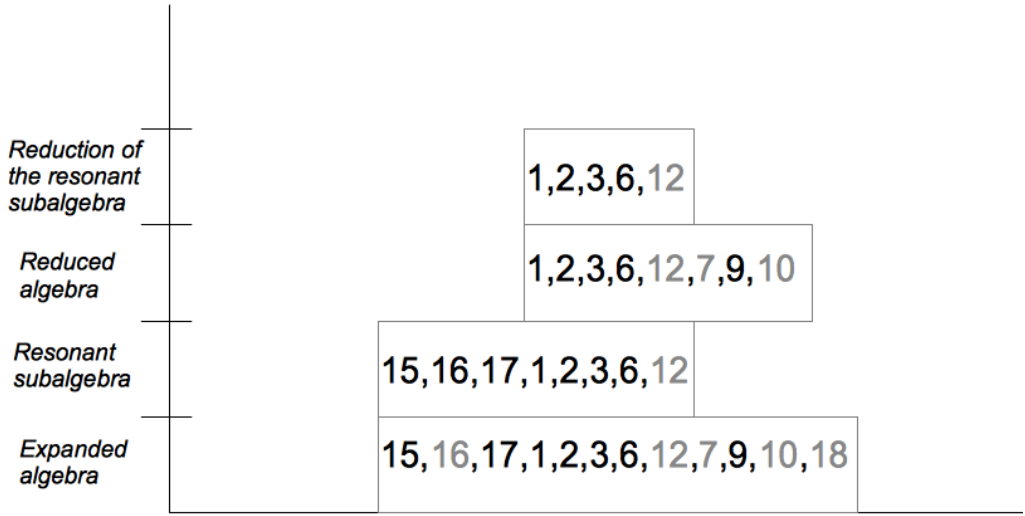}\caption{Expansions of
$\mathfrak{sl}(2,\mathbb{R})$ with abelian semigroups of order $3$}%
\label{fig:fig0}%
\end{figure}

We identified the different semigroups $S_{(3)}^{a} $ of order 3 with a label
`$a$' for reasons of space in figure \ref{fig:fig0} to name the different
semigroups. The horizontal axis represents the set of semigroups used in some
specific expansion while the vertical axis represents the different kinds of
expansions that can be performed. So,

\begin{itemize}
\item in the first level we list all the abelian semigroups that allow us to
perform a general expansion $S\otimes\mathcal{G}$. The abelian
semigroups are: $S_{\left(  3\right)  }^{1}$, $S_{\left(  3\right)  }^{2}$,
$S_{\left(  3\right)  }^{3}$, $S_{\left(  3\right)  }^{6}$, $S_{\left(
3\right)  }^{7}$, $S_{\left(  3\right)  }^{9}$, $S_{\left(  3\right)  }^{10}$,
$S_{\left(  3\right)  }^{12}$, $S_{\left(  3\right)  }^{15}$, $S_{\left(
3\right)  }^{16}$, $S_{\left(  3\right)  }^{17}$, $S_{\left(  3\right)  }%
^{18}$.

\item in the second level we find all the abelian semigroups that contain at
least one resonant decomposition so that a resonant subalgebra can be
extracted from the expanded one. The abelian semigroups are: $S_{\left(
3\right)  }^{1}$, $S_{\left(  3\right)  }^{2}$, $S_{\left(  3\right)  }^{3}$,
$S_{\left(  3\right)  }^{6}$, $S_{\left(  3\right)  }^{12}$, $S_{\left(
3\right)  }^{15}$, $S_{\left(  3\right)  }^{16}$, $S_{\left(  3\right)  }%
^{17}$.

\item in the third level we see all the abelian semigroups that contain a zero
element so that a reduced algebra can be extracted from the expanded one.
Here are the semigroups: $S_{\left(  3\right)  }^{1}$, $S_{\left(  3\right)
}^{2}$, $S_{\left(  3\right)  }^{3}$, $S_{\left(  3\right)  }^{6}$,
$S_{\left(  3\right)  }^{7}$, $S_{\left(  3\right)  }^{9}$, $S_{\left(
3\right)  }^{10}$, $S_{\left(  3\right)  }^{12}$.

\item in the fourth level we find all the abelian semigroups that contain at
least one resonant decomposition and also a zero element. So a reduced algebra
can be obtained from the resonant subalgebra. The semigroups that allow us to
do that are: $S_{\left(  3\right)  }^{1}$, $S_{\left(  3\right)  }^{2}$,
$S_{\left(  3\right)  }^{3}$, $S_{\left(  3\right)  }^{6}$, $S_{\left(
3\right)  }^{12}$.
\end{itemize}

Then, for all these expansions, we identify the semigroups that preserve the
semisimplicity of the original algebra. In the graphic of figure
\ref{fig:fig0} those semigroups are labeled with a \textit{gray number}.

As we know, the $S$-expansion method \cite{irs} contains as a particular case
the expansions by means of a parameter \cite{aipv1} (and at the same time this
method includes all the types of contractions that are known in the
literature). In fact, this kind of expansion is recovered when one specific semigroup\footnote{In the
case of expansions with semigroups of order 3 this happens for the semigroup
$S_{\left(  3\right)  }^{6}$\ given by
\[%
\begin{tabular}
[c]{l|lll}%
$S_{\left(  3\right)  }^{6}$ & $\lambda_{1}$ & $\lambda_{2}$ & $\lambda_{3}%
$\\\hline
$\lambda_{1}$ & $\lambda_{1}$ & $\lambda_{1}$ & $\lambda_{1}$\\
$\lambda_{2}$ & $\lambda_{1}$ & $\lambda_{1}$ & $\lambda_{2}$\\
$\lambda_{3}$ & $\lambda_{1}$ & $\lambda_{2}$ & $\lambda_{3}$%
\end{tabular}
\ \ \ \ \ \ \ \
\]
which is isomorphic to the semigroup $S_{E}^{\left(  n\right)  }$ with $n=1$
that is used in the original article \cite{irs} where $S$-expansions where
introduced.}, in the fourth level of figure \ref{fig:fig0}, is used. So, it is
evident that the use of semigroups leads to more general kinds of expansions,
i.e., to expansions that are not only on the fourth level of figure
\ref{fig:fig0}, but also in the other three lower levels. Besides, in each of
them there are cases where semisimplicity can or cannot be preserved.

\subsection{General properties of the expansions with $n=3,...,6$}

For higher orders $n\geq4$ it is not possible to show a graphic like that
given in figure \ref{fig:fig0}. Instead we give a table listing the number of
semigroups that leads to the different kinds of expansions we have mentioned
in the previous section. We also give in each case the number of semigroups
preserving semisimplicity:%

\begin{equation}%
\begin{tabular}
[c]{|l|l|l|l|l|}\hline
order & 3 & 4 & 5 & 6\\\hline%
\begin{tabular}
[c]{|l|}\hline%
\begin{tabular}
[c]{l}%
{\small expanded}\\
${\small S\otimes}\mathcal{G}$%
\end{tabular}
\\\hline%
\begin{tabular}
[c]{l}%
{\small preserving}\\
{\small semisimplicity}%
\end{tabular}
\\\hline
\end{tabular}
&
\begin{tabular}
[c]{|l|}\hline
\#12\\\hline
\#5\\\hline
\end{tabular}
&
\begin{tabular}
[c]{|l|}\hline
\#58\\\hline
\#16\\\hline
\end{tabular}
&
\begin{tabular}
[c]{|l|}\hline
\#325\\\hline
\#51\\\hline
\end{tabular}
&
\begin{tabular}
[c]{|l|}\hline
\#2,143\\\hline
\#201\\\hline
\end{tabular}
\\\hline%
\begin{tabular}
[c]{|l|}\hline%
\begin{tabular}
[c]{l}%
{\small expanded}\\
{\small and reduced}%
\end{tabular}
\\\hline%
\begin{tabular}
[c]{l}%
{\small preserving}\\
{\small semisimplicity}%
\end{tabular}
\\\hline
\end{tabular}
&
\begin{tabular}
[c]{|l|}\hline
\#8\\\hline
\#3\\\hline
\end{tabular}
&
\begin{tabular}
[c]{|l|}\hline
\#39\\\hline
\#9\\\hline
\end{tabular}
&
\begin{tabular}
[c]{|l|}\hline
\#226\\\hline
\#34\\\hline
\end{tabular}
&
\begin{tabular}
[c]{|l|}\hline
\#1,538\\\hline
\#135\\\hline
\end{tabular}
\\\hline%
\begin{tabular}
[c]{|l|}\hline%
\begin{tabular}
[c]{l}%
{\small resonant}\\
{\small subalgebra}%
\end{tabular}
\\\hline%
\begin{tabular}
[c]{l}%
{\small preserving}\\
{\small semisimplicity}%
\end{tabular}
\\\hline
\end{tabular}
&
\begin{tabular}
[c]{|l|}\hline
\#8\\\hline
\#1\\\hline
\end{tabular}
&
\begin{tabular}
[c]{|l|}\hline
\#48\\\hline
\#4\\\hline
\end{tabular}
&
\begin{tabular}
[c]{|l|}\hline
\#299\\\hline
\#7\\\hline
\end{tabular}
&
\begin{tabular}
[c]{|l|}\hline
\#2,059\\\hline
\#23\\\hline
\end{tabular}
\\\hline%
\begin{tabular}
[c]{|l|}\hline%
\begin{tabular}
[c]{l}%
{\small reduction of}\\
{\small resonant subalg.}%
\end{tabular}
\\\hline%
\begin{tabular}
[c]{l}%
{\small preserving}\\
{\small semisimplicity}%
\end{tabular}
\\\hline
\end{tabular}
&
\begin{tabular}
[c]{|l|}\hline
\#5\\\hline
\#1\\\hline
\end{tabular}
&
\begin{tabular}
[c]{|l|}\hline
\#32\\\hline
\#1\\\hline
\end{tabular}
&
\begin{tabular}
[c]{|l|}\hline
\#204\\\hline
\#6\\\hline
\end{tabular}
&
\begin{tabular}
[c]{|l|}\hline
\#1,465\\\hline
\#12\\\hline
\end{tabular}
\\\hline
\end{tabular}
\ \ \ \label{cp_table1}%
\end{equation}

To construct this table we have written a set of computer programs (a Java Library) that take as an input all the abelian semigroups given by the the second program in \cite{Plemmons}, \textit{com.f}, and then permit to classify those ones with a zero elements and also to study all the resonant decompositions they have. A deeper description about those programs submitted soon (see \cite{Nadal}).

In the different rows we see the various kinds of expansions that can be done
(the expanded algebra, the resonant subalgebra, the reduced algebra and the
reduction of the resonant subalgebra) for each order of the semigroups. In
each case the number of semigroups with which the expansions can be performed
is specified. The number of semigroups preserving semi-simplicity is also
given. However, it may happen that a certain semigroup has more than one
resonance, which then leads to different expanded algebras. In what follows we
summarize this information:

\textbf{For the order }$n=3$\textbf{:}

\begin{itemize}
\item There are $8$ semigroups of order $3$ with at least one resonant decomposition,

\item with $9$ different resonant decompositions (so this is the number of
different kinds of expansions that can be made for $n=3$)

\item and $1$ expansion that gives a semisimple Lie algebra with one of its resonances.
\end{itemize}

\textbf{For the order }$n=4$\textbf{:}

\begin{itemize}
\item There are $48$ semigroups of order $4$ with at least one resonant decomposition,

\item with $124$ different resonant decompositions (so this is the number of
different kinds of expansions that can be made for $n=4$)

\item and $4$ expansions that give a semi-simple Lie algebra with one of its resonances.
\end{itemize}

\textbf{For the order }$n=5$\textbf{:}

\begin{itemize}
\item There are $299$ semigroups of order $5$ with at least one resonant decomposition,

\item with $1,653$ different resonant decompositions (so this is the number of
different kinds of expansions that can be made for $n=5$)

\item and $7$ expansions that give a semi-simple Lie algebra with one of its resonances.
\end{itemize}

\textbf{For the order }$n=6$\textbf{:}

\begin{itemize}
\item There are $2,059$ semigroups of order $6$ with at least one resonant decomposition,

\item with $25,512$ different resonant decompositions (so this is the number
of different kinds of expansions that can be made for $n=6$)

\item and $23$ expansions that give a semi-simple Lie algebra with one of its resonances.
\end{itemize}

In general all these expansions with $n=3,4,5,6$ share the following property:

\begin{remark}
Consider a semigroup of order $n=3,4,5,6$ having more than one resonance. Then
if it preserves semisimplicity, this happens just for one of its resonances.
There is no semigroup preserving semisimplicity with more than one of its resonances.
\end{remark}

In this way it is explicitly verified that starting from a semi-simple algebra
the expanded algebras are not necesarily semi-simple, as was suggested in
section \ref{preserved}. In fact the major part of the expansions do not
preserve semi-simplicity.


\section{Semigroups preserving semisimplicity}

\label{sclassification}

We have seen that semi-simplicity and compactness are preserved only in
expansions performed with semigroups satisfying certain conditions related to
the signature (i.e., number of positive, negative and zero eigenvalues) of
the symmetric matrices $\mathbf{g}^{S}$, $\mathbf{g}^{S}\left(  S_{0}\right)
$, $\mathbf{g}^{S}\left(  S_{1}\right)  $, $\mathbf{g}^{\text{red}}\left(
S_{0}\right)  $ and $\mathbf{g}^{\text{red}}\left(  S_{0}\right)  $ defined in
section \ref{semis}. It is therefore relevant to classify the semigroups also
in relation to that property, which is done here for semigroups of order $3,4$
and $5$ preserving semisimplicity.

Using the standard classification and notation of the semigroups given in
section \ref{history} we have found that the semigroups $S_{\left(  n\right)
}^{a}=\left\{  \lambda_{1},...,\lambda_{n}\right\}  $ of order $n=3,4,5$\ with
a zero element and a resonant decomposition (\ref{lres9}) that preserve
semisimplicity are: $S_{\left(  3\right)  }^{12}$, $S_{\left(  4\right)
}^{88}$, $S_{\left(  5\right)  }^{770}$, $S_{\left(  5\right)  }^{779}$,
$S_{\left(  5\right)  }^{922}$, $S_{\left(  5\right)  }^{968}$, $S_{\left(
5\right)  }^{990}$, $S_{\left(  5\right)  }^{991}$. In what follow we list
their main properties and explicit expansions of $\mathfrak{sl}(2,\mathbb{R})$
will be performed with those semigroups:

The semigroup $S_{\left(  3\right)  }^{12}=\left\{  \lambda_{1},\lambda
_{2},\lambda_{3}\right\}  $ is defined by%

\[%
\begin{tabular}
[c]{l|lll}%
$S_{\left(  3\right)  }^{12}$ & $\lambda_{1}$ & $\lambda_{2}$ & $\lambda_{3}%
$\\\hline
$\lambda_{1}$ & $\lambda_{1}$ & $\lambda_{1}$ & $\lambda_{1}$\\
$\lambda_{2}$ & $\lambda_{1}$ & $\lambda_{2}$ & $\lambda_{3}$\\
$\lambda_{3}$ & $\lambda_{1}$ & $\lambda_{3}$ & $\lambda_{2}$%
\end{tabular}
\ \ \ \ \ \text{ \
\begin{tabular}
[c]{l}%
with the resonant decomposition\\
\multicolumn{1}{c}{$S_{0}=\left\{  \lambda_{1},\lambda_{2}\right\}  $,
$S_{1}=\left\{  \lambda_{1},\lambda_{3}\right\}  $}%
\end{tabular}
}%
\]
The eigenvalues of the matrices $\mathbf{g}^{S}$, $\mathbf{g}^{S}\left(
S_{0}\right)  $, $\mathbf{g}^{S}\left(  S_{1}\right)  $, $\mathbf{g}%
^{\text{red}}\left(  S_{0}\right)  $, $\mathbf{g}^{\text{red}}\left(
S_{0}\right)  $ for this semigroup are all positive so, besides
semisimplicity, compactness properties are also preserved. However, an
expansion with this semigroup is trivial in the sense that the reduction of
the resonant subalgebra is equal to the original algebra. Something similar
happens for the semigroup $S^{88}=\left\{  \lambda_{1},\lambda_{2},\lambda
_{3},\lambda_{4}\right\}  $ defined by%
\[%
\begin{tabular}
[c]{l|llll}%
$S_{\left(  4\right)  }^{88}$ & $\lambda_{1}$ & $\lambda_{2}$ & $\lambda_{3}$
& $\lambda_{4}$\\\hline
$\lambda_{1}$ & $\lambda_{1}$ & $\lambda_{1}$ & $\lambda_{1}$ & $\lambda_{1}%
$\\
$\lambda_{2}$ & $\lambda_{1}$ & $\lambda_{2}$ & $\lambda_{2}$ & $\lambda_{2}%
$\\
$\lambda_{3}$ & $\lambda_{1}$ & $\lambda_{2}$ & $\lambda_{3}$ & $\lambda_{4}%
$\\
$\lambda_{4}$ & $\lambda_{1}$ & $\lambda_{2}$ & $\lambda_{4}$ & $\lambda_{4}$%
\end{tabular}
\ \ \ \ \ \text{ \
\begin{tabular}
[c]{l}%
with the resonant decomposition\\
\multicolumn{1}{c}{$S_{0}=\left\{  \lambda_{1},\lambda_{2},\lambda
_{3}\right\}  $, $S_{1}=\left\{  \lambda_{1},\lambda_{2},\lambda_{4}\right\}
$.}%
\end{tabular}
}%
\]
It also preserves compactness properties, but two consecutive
reductions\footnote{In general more than one reduction can be performed when a
semigroup (here $S_{\left(  4\right)  }^{88}$) with a certain zero element
($\lambda_{1}$) contains a sub-semigroup (the set $\left\{  \lambda
_{2},\lambda_{3},\lambda_{4}\right\}  $) having its own zero element (in this
case $\lambda_{2}$).} of the resonant subalgebra lead finally to the original
algebra, so it is not interesting either.

More interesting cases can be found among the $6$ semigroups of order $5$ that
preserve semisimplicity of the reduction of the resonant subalgebra (see table
\ref{cp_table1}). Those semigroups are:%

\begin{equation}%
\begin{tabular}
[c]{l}%
$%
\begin{tabular}
[c]{c|ccccc}%
$S_{\left(  5\right)  }^{770}$ & $\lambda_{1}$ & $\lambda_{2}$ & $\lambda_{3}$
& $\lambda_{4}$ & $\lambda_{5}$\\\hline
$\lambda_{1}$ & $\lambda_{1}$ & $\lambda_{1}$ & $\lambda_{1}$ & $\lambda_{1}$
& $\lambda_{1}$\\
$\lambda_{2}$ & $\lambda_{1}$ & $\lambda_{2}$ & $\lambda_{1}$ & $\lambda_{1}$
& $\lambda_{5}$\\
$\lambda_{3}$ & $\lambda_{1}$ & $\lambda_{1}$ & $\lambda_{3}$ & $\lambda_{4}$
& $\lambda_{1}$\\
$\lambda_{4}$ & $\lambda_{1}$ & $\lambda_{1}$ & $\lambda_{4}$ & $\lambda_{3}$
& $\lambda_{1}$\\
$\lambda_{5}$ & $\lambda_{1}$ & $\lambda_{5}$ & $\lambda_{1}$ & $\lambda_{1}$
& $\lambda_{2}$%
\end{tabular}
\ \ \ \ $\\
with resonant decomposition:\\
\multicolumn{1}{c}{$S_{0}=\left\{  \lambda_{1},\lambda_{2},\lambda
_{3}\right\}  $}\\
\multicolumn{1}{c}{$S_{1}=\left\{  \lambda_{1},\lambda_{4},\lambda
_{5}\right\}  $}%
\end{tabular}
\ \ \ \ \ \ \ \
\begin{tabular}
[c]{l}%
$%
\begin{tabular}
[c]{c|ccccc}%
$S_{\left(  5\right)  }^{779}$ & $\lambda_{1}$ & $\lambda_{2}$ & $\lambda_{3}$
& $\lambda_{4}$ & $\lambda_{5}$\\\hline
$\lambda_{1}$ & $\lambda_{1}$ & $\lambda_{1}$ & $\lambda_{1}$ & $\lambda_{1}$
& $\lambda_{1}$\\
$\lambda_{2}$ & $\lambda_{1}$ & $\lambda_{2}$ & $\lambda_{1}$ & $\lambda_{2}$
& $\lambda_{2}$\\
$\lambda_{3}$ & $\lambda_{1}$ & $\lambda_{1}$ & $\lambda_{3}$ & $\lambda_{3}$
& $\lambda_{3}$\\
$\lambda_{4}$ & $\lambda_{1}$ & $\lambda_{2}$ & $\lambda_{3}$ & $\lambda_{4}$
& $\lambda_{5}$\\
$\lambda_{5}$ & $\lambda_{1}$ & $\lambda_{2}$ & $\lambda_{3}$ & $\lambda_{5}$
& $\lambda_{4}$%
\end{tabular}
\ \ \ \ $\\
with resonant decomposition:\\
\multicolumn{1}{c}{$S_{0}=\left\{  \lambda_{1},\lambda_{2},\lambda_{3}%
,\lambda_{4}\right\}  $}\\
\multicolumn{1}{c}{$S_{1}=\left\{  \lambda_{1},\lambda_{2},\lambda_{3}%
,\lambda_{5}\right\}  $}%
\end{tabular}
\ \ \ \ \label{y1}%
\end{equation}

\begin{equation}%
\begin{tabular}
[c]{l}%
$%
\begin{tabular}
[c]{c|ccccc}%
$S_{\left(  5\right)  }^{922}$ & $\lambda_{1}$ & $\lambda_{2}$ & $\lambda_{3}$
& $\lambda_{4}$ & $\lambda_{5}$\\\hline
$\lambda_{1}$ & $\lambda_{1}$ & $\lambda_{1}$ & $\lambda_{1}$ & $\lambda_{1}$
& $\lambda_{1}$\\
$\lambda_{2}$ & $\lambda_{1}$ & $\lambda_{2}$ & $\lambda_{2}$ & $\lambda_{2}$
& $\lambda_{2}$\\
$\lambda_{3}$ & $\lambda_{1}$ & $\lambda_{2}$ & $\lambda_{3}$ & $\lambda_{3}$
& $\lambda_{3}$\\
$\lambda_{4}$ & $\lambda_{1}$ & $\lambda_{2}$ & $\lambda_{3}$ & $\lambda_{4}$
& $\lambda_{5}$\\
$\lambda_{5}$ & $\lambda_{1}$ & $\lambda_{2}$ & $\lambda_{3}$ & $\lambda_{5}$
& $\lambda_{4}$%
\end{tabular}
\ \ \ $\\
with resonant decomposition:\\
\multicolumn{1}{c}{$S_{0}=\left\{  \lambda_{1},\lambda_{2},\lambda_{3}%
,\lambda_{4}\right\}  $}\\
\multicolumn{1}{c}{$S_{1}=\left\{  \lambda_{1},\lambda_{2},\lambda_{3}%
,\lambda_{5}\right\}  $}%
\end{tabular}
\ \ \ \ \ \ \
\begin{tabular}
[c]{l}%
$%
\begin{tabular}
[c]{c|ccccc}%
$S_{\left(  5\right)  }^{968}$ & $\lambda_{1}$ & $\lambda_{2}$ & $\lambda_{3}$
& $\lambda_{4}$ & $\lambda_{5}$\\\hline
$\lambda_{1}$ & $\lambda_{1}$ & $\lambda_{1}$ & $\lambda_{1}$ & $\lambda_{1}$
& $\lambda_{1}$\\
$\lambda_{2}$ & $\lambda_{1}$ & $\lambda_{2}$ & $\lambda_{2}$ & $\lambda_{4}$
& $\lambda_{4}$\\
$\lambda_{3}$ & $\lambda_{1}$ & $\lambda_{2}$ & $\lambda_{3}$ & $\lambda_{4}$
& $\lambda_{5}$\\
$\lambda_{4}$ & $\lambda_{1}$ & $\lambda_{4}$ & $\lambda_{4}$ & $\lambda_{2}$
& $\lambda_{2}$\\
$\lambda_{5}$ & $\lambda_{1}$ & $\lambda_{4}$ & $\lambda_{5}$ & $\lambda_{2}$
& $\lambda_{3}$%
\end{tabular}
\ \ \ $\\
with resonant decomposition:\\
\multicolumn{1}{c}{$S_{0}=\left\{  \lambda_{1},\lambda_{2},\lambda
_{3}\right\}  $}\\
\multicolumn{1}{c}{$S_{1}=\left\{  \lambda_{1},\lambda_{4},\lambda
_{5}\right\}  $}%
\end{tabular}
\ \ \ \label{y2}%
\end{equation}

\begin{equation}%
\begin{tabular}
[c]{l}%
$%
\begin{tabular}
[c]{c|ccccc}%
$S_{\left(  5\right)  }^{990}$ & $\lambda_{1}$ & $\lambda_{2}$ & $\lambda_{3}$
& $\lambda_{4}$ & $\lambda_{5}$\\\hline
$\lambda_{1}$ & $\lambda_{1}$ & $\lambda_{1}$ & $\lambda_{1}$ & $\lambda_{1}$
& $\lambda_{1}$\\
$\lambda_{2}$ & $\lambda_{1}$ & $\lambda_{2}$ & $\lambda_{3}$ & $\lambda_{4}$
& $\lambda_{5}$\\
$\lambda_{3}$ & $\lambda_{1}$ & $\lambda_{3}$ & $\lambda_{2}$ & $\lambda_{5}$
& $\lambda_{4}$\\
$\lambda_{4}$ & $\lambda_{1}$ & $\lambda_{4}$ & $\lambda_{5}$ & $\lambda_{2}$
& $\lambda_{3}$\\
$\lambda_{5}$ & $\lambda_{1}$ & $\lambda_{5}$ & $\lambda_{4}$ & $\lambda_{3}$
& $\lambda_{2}$%
\end{tabular}
\ \ \ \ \ $\\
with resonant decomposition:\\
\multicolumn{1}{c}{$S_{0}=\left\{  \lambda_{1},\lambda_{2},\lambda
_{5}\right\}  $}\\
\multicolumn{1}{c}{$S_{1}=\left\{  \lambda_{1},\lambda_{3},\lambda
_{4}\right\}  $}%
\end{tabular}
\ \ \ \ \ \ \ \ \
\begin{tabular}
[c]{l}%
$%
\begin{tabular}
[c]{c|ccccc}%
$S_{\left(  5\right)  }^{991}$ & $\lambda_{1}$ & $\lambda_{2}$ & $\lambda_{3}$
& $\lambda_{4}$ & $\lambda_{5}$\\\hline
$\lambda_{1}$ & $\lambda_{1}$ & $\lambda_{1}$ & $\lambda_{1}$ & $\lambda_{1}$
& $\lambda_{1}$\\
$\lambda_{2}$ & $\lambda_{1}$ & $\lambda_{2}$ & $\lambda_{3}$ & $\lambda_{4}$
& $\lambda_{5}$\\
$\lambda_{3}$ & $\lambda_{1}$ & $\lambda_{3}$ & $\lambda_{2}$ & $\lambda_{5}$
& $\lambda_{4}$\\
$\lambda_{4}$ & $\lambda_{1}$ & $\lambda_{4}$ & $\lambda_{5}$ & $\lambda_{3}$
& $\lambda_{2}$\\
$\lambda_{5}$ & $\lambda_{1}$ & $\lambda_{5}$ & $\lambda_{4}$ & $\lambda_{2}$
& $\lambda_{3}$%
\end{tabular}
\ \ \ \ \ $\\
with resonant decomposition:\\
\multicolumn{1}{c}{$S_{0}=\left\{  \lambda_{1},\lambda_{2},\lambda
_{3}\right\}  $}\\
\multicolumn{1}{c}{$S_{1}=\left\{  \lambda_{1},\lambda_{4},\lambda
_{5}\right\}  $}%
\end{tabular}
\ \ \ \ \ \label{y3}%
\end{equation}
\bigskip

The semigroups $S_{\left(  5\right)  }^{779}$ and $S_{\left(  5\right)
}^{922}$ also lead to trivial expansions: three reductions can be performed on
the resonant subalgebra leading finally to the original algebra\footnote{See
the structure of their multiplication tables and remember the comment made about
semigroups containing sub-semigroups with its own zero elements. In some sense
these semigroups can be considered as an enlargement of $S_{\left(  3\right)
}^{12}$ by new zero elements.}. However, more interesting structures are
obtained with the other semigroups.

\subsection{The semigroup $S_{\left(  5\right)  }^{770}$}

Expanding $\mathfrak{sl}\left(  2,\mathbb{R}\right)  $ with the semigroup
$S_{\left(  5\right)  }^{770}$, the reduction of the resonant subalgebra is
given by%
\begin{align}
\mathcal{G}_{0,S,R}^{\text{red}} &  =\left(  S_{0}^{\text{red}}\otimes
\mathcal{T}_{0}\right)  +\left(  S_{1}^{\text{red}}\otimes\mathcal{P}%
_{0}\right)  \label{GS770}\\
&  =\left\{  X_{\left(  1,2\right)  },X_{\left(  1,3\right)  }\right\}
+\left\{  X_{\left(  2,4\right)  },X_{\left(  2,5\right)  },X_{\left(
3,4\right)  },X_{\left(  3,5\right)  },\right\}  \nonumber\\
&  =\left\{  \lambda_{2}\otimes\left(  -i\sigma_{2}\right)  ,\ \left(
\lambda_{3}\otimes-i\sigma_{2}\right)  \right\}  +\left\{  \lambda_{4}%
\otimes\sigma_{1},\ \lambda_{5}\otimes\sigma_{1},\ \lambda_{4}\otimes
\sigma_{3},\ \lambda_{5}\otimes\sigma_{3}\right\}  \nonumber
\end{align}
where we have used the notation given in (\ref{not1}), (\ref{not2}). Renaming the generators as%
\begin{equation}%
\begin{tabular}
[c]{lll}%
$Y_{1}=X_{\left(  1,2\right)  }$ &  & $Y_{4}=X_{\left(  2,5\right)  }$\\
$Y_{2}=X_{\left(  1,3\right)  }$ &  & $Y_{5}=X_{\left(  3,4\right)  }$\\
$Y_{3}=X_{\left(  2,4\right)  }$ &  & $Y_{6}=X_{\left(  3,5\right)  }$.
\end{tabular}
\label{GS770__1}%
\end{equation}

Using the equation (\ref{z3}), the commutation relations of $\mathfrak{sl}%
(2,\mathbb{R})$ and the multiplication table of $S_{\left(  5\right)  }^{770}$
(given in (\ref{y1})) we obtain the following commutation relations for
$\mathcal{G}_{0,S,R}^{\text{red}}$%
\begin{equation}%
\begin{tabular}
[c]{lll}%
$\left[  Y_{1},Y_{3}\right]  =0$ &  & $\left[  Y_{2},Y_{5}\right]  =2Y_{3}$\\
$\left[  Y_{1},Y_{4}\right]  =-2Y_{6}$ &  & $\left[  Y_{2},Y_{6}\right]  =0$\\
$\left[  Y_{1},Y_{5}\right]  =0$ &  & $\left[  Y_{3},Y_{5}\right]  =2Y_{2}$\\
$\left[  Y_{1},Y_{6}\right]  =2Y_{4}$ &  & $\left[  Y_{3},Y_{6}\right]  =0$\\
$\left[  Y_{2},Y_{3}\right]  =-2Y_{5}$ &  & $\left[  Y_{4},Y_{5}\right]  =0$\\
$\left[  Y_{2},Y_{4}\right]  =0$ &  & $\left[  Y_{4},Y_{6}\right]  =2Y_{1}$%
\end{tabular}
\label{GS770__2}%
\end{equation}

It is direct to see that this algebra is just $\mathfrak{sl}(2,\mathbb{R}%
)\oplus\mathfrak{sl}(2,\mathbb{R})$, which is obviously semisimple. The two
commuting $\mathfrak{sl}\left(  2,\mathbb{R}\right)  $ algebras are
respectively generated by $\left\{  Y_{1},Y_{4},Y_{6}\right\}  $ and $\left\{
Y_{2},Y_{3},Y_{5}\right\}  $. The Killing-Cartan metric, calculated with those
structure constants, is given by:%
\begin{equation}
\mathbf{g}^{E,R,\text{red}}=diag\left(  -8,-8,8,8,8,8\right)  \label{n1_GS770}%
\end{equation}

As shown in equation (\ref{res_red_final}) of section \ref{semis} the
semisimplicity and compactness properties of the expanded algebra can be
directly predicted from intrinsic properties of the semigroup used in the
expansion. In fact, the matrices $\mathbf{g}^{\text{red}}\left(  S_{0}\right)
$, $\mathbf{g}^{\text{red}}\left(  S_{1}\right)  $ can be directly calculated
from the multiplication table of the semigroup $S_{\left(  5\right)  }^{770}$%
\begin{equation}
\mathbf{g}^{\text{red}}\left(  S_{0}\right)  =\mathbf{g}^{\text{red}}\left(
S_{1}\right)  =\left(
\begin{array}
[c]{cc}%
2 & 0\\
0 & 2
\end{array}
\right)  \text{.} \label{n2_GS770}%
\end{equation}
From eq. (\ref{res_red_final}) and since the eigenvalues $\xi_{\alpha_{0}}$
and $\xi_{\alpha_{1}}$ of $\mathbf{g}^{\text{red}}\left(  S_{0}\right)  $ and
$\mathbf{g}^{\text{red}}\left(  S_{1}\right)  $ are all positive we conclude
that the expansion of whatever algebra, with a decomposition $\mathcal{G}%
=V_{0}\oplus V_{1}$ and having the subspace structure (\ref{lres8}), will
preserve semisimplicity and compactness properties. In particular, if
$\mathcal{G}=\mathfrak{sl}\left(  2,\mathbb{R}\right)  $ (with $V_{0}%
=\mathcal{T}_{0}$ and $V_{1}=\mathcal{P}_{0}$) then from (\ref{res_red_final},
\ref{not05}, \ref{n1_GS770}) we obtain $\alpha=1/2$. Besides%
\begin{equation}
\mathcal{G}_{0,S,R}^{\text{red}}=\mathcal{T}_{0,S,R}^{\text{red}}%
+\mathcal{P}_{0,S,R}^{\text{red}} \label{n5GS770}%
\end{equation}
where%
\begin{align}
\mathcal{T}_{0,S,R}^{\text{red}}  &  =\mathcal{T}_{0}\otimes S_{0}%
^{\text{red}}\label{n6GS770}\\
\mathcal{P}_{0,S,R}^{\text{red}}  &  =\mathcal{P}_{0}\otimes S_{1}%
^{\text{red}}\nonumber
\end{align}
is the Cartan decomposition for the algebra $\mathcal{G}_{0,S,R}^{\text{red}}%
$, while the real compact form of $\mathcal{G}_{S,R}^{\text{red}}$ (the
complex form of $\mathcal{G}_{0,S,R}^{\text{red}}$) is given by
\begin{align}
\mathcal{G}_{k,S,R}^{\text{red}}  &  =\left(  S_{0}^{\text{red}}%
\otimes\mathcal{T}_{0}\right)  +\left(  S_{1}^{\text{red}}\otimes
i\mathcal{P}_{0}\right) \label{GS700_1}\\
&  =\left\{  \lambda_{2}\otimes i\sigma_{2},\ \lambda_{3}\otimes i\sigma
_{2}\right\}  +\left\{  \lambda_{4}\otimes i\sigma_{1},\ \lambda_{5}\otimes
i\sigma_{1},\ \lambda_{4}\otimes i\sigma_{3},\ \lambda_{5}\otimes i\sigma
_{3}\right\}  \text{.}\nonumber
\end{align}
Finally, relations (\ref{f7}) can also be directly checked,%
\begin{align}
\mathcal{T}_{0,S,R}^{\text{red}}  &  =\mathcal{G}_{0,S,R}^{\text{red}}%
\cap\mathcal{G}_{k,S,R}^{\text{red}}=\left\{  \lambda_{2}\otimes\left(
-i\sigma_{2}\right)  ,\ \lambda_{3}\otimes\left(  -i\sigma_{2}\right)
\right\} \label{GS770_2}\\
\mathcal{P}_{0,S,R}^{\text{red}}  &  =\mathcal{G}_{0,S,R}^{\text{red}}%
\cap\left(  i\mathcal{G}_{k,S,R}^{\text{red}}\right)  =\left\{  \lambda
_{4}\otimes\sigma_{1},\ \lambda_{5}\otimes\sigma_{1},\ \lambda_{4}%
\otimes\sigma_{3},\ \lambda_{5}\otimes\sigma_{3}\right\}  \text{.}\nonumber
\end{align}

\subsection{The semigroup $S_{\left(  5\right)  }^{968}$}

For this semigroup the matrices $\mathbf{g}^{\text{red}}\left(  S_{0}\right)
$, $\mathbf{g}^{\text{red}}\left(  S_{1}\right)  $ are given by:%
\begin{equation}
\mathbf{g}^{\text{red}}\left(  S_{0}\right)  =\mathbf{g}^{\text{red}}\left(
S_{1}\right)  =\left(
\begin{array}
[c]{cc}%
2 & 2\\
2 & 4
\end{array}
\right)  \text{.} \label{n1GS968}%
\end{equation}
Since their eigenvalues are all positive we conclude that semisimplicity and
compactness properties should be preserved. Considering $\mathcal{G}%
=\mathfrak{sl}\left(  2,\mathbb{R}\right)  $ equation (\ref{res_red_final})
takes the form%
\begin{equation}
\mathbf{g}^{E,R,\text{red}}=\alpha\left(
\begin{array}
[c]{cc}%
\left(
\begin{array}
[c]{cc}%
2 & 2\\
2 & 4
\end{array}
\right)  \otimes\left(  -8\right)  & 0\\
0 & \left(
\begin{array}
[c]{cc}%
2 & 2\\
2 & 4
\end{array}
\right)  \otimes\left(
\begin{array}
[c]{cc}%
8 & 0\\
0 & 8
\end{array}
\right)
\end{array}
\right)  \text{,} \label{n2GS968}%
\end{equation}
from where we can predict that the expanded algebra will have $2$ compact
generators and $4$ non-compact generators. In what follows, we are going to show that the expanded algebra is $\mathfrak{sl}%
\left(  2,\mathbb{R}\right)  \oplus\mathfrak{sl}\left(  2,\mathbb{R}\right)
$.

Expanding with the semigroup $S_{\left(  5\right)  }^{968}$ we have,%
\begin{align}
\mathcal{G}_{0,S,R}^{\text{red}} &  =\left(  S_{0}^{\text{red}}\otimes
\mathcal{T}_{0}\right)  +\left(  S_{1}^{\text{red}}\otimes\mathcal{P}%
_{0}\right)  \label{GS968}\\
&  =\left\{  X_{\left(  1,2\right)  },X_{\left(  1,3\right)  }\right\}
+\left\{  X_{\left(  2,4\right)  },X_{\left(  2,5\right)  },X_{\left(
3,4\right)  },X_{\left(  3,5\right)  },\right\}  \nonumber\\
&  =\left\{  \lambda_{2}\otimes\left(  -i\sigma_{2}\right)  ,\ \lambda
_{3}\otimes\left(  -i\sigma_{2}\right)  \right\}  +\left\{  \lambda_{4}%
\otimes\sigma_{1},\ \lambda_{5}\otimes\sigma_{1},\ \lambda_{4}\otimes
\sigma_{3},\ \lambda_{5}\otimes\sigma_{3}\right\}  \nonumber
\end{align}
Note that the structure is very similar to the algebra (\ref{GS770}). What
makes the difference between those algebras is that the elements of the
semigroups obey different tables of multiplication. Let us also define
$Y_{1},Y_{2},...,Y_{6}$ as in equation (\ref{GS770__1}). Using the equation
(\ref{z3}), the commutation relations of $\mathfrak{sl}(2,\mathbb{R})$ and the
multiplication table of $S_{\left(  5\right)  }^{968}$ (given in (\ref{y2}))
we obtain the following commutation relations for $\mathcal{G}_{0,S,R}%
^{\text{red}}$
\begin{equation}%
\begin{tabular}
[c]{lll}%
$\left[  Y_{1},Y_{3}\right]  =-2Y_{5}$ &  & $\left[  Y_{2},Y_{5}\right]
=2Y_{3}$\\
$\left[  Y_{1},Y_{4}\right]  =-2Y_{5}$ &  & $\left[  Y_{2},Y_{6}\right]
=2Y_{4}$\\
$\left[  Y_{1},Y_{5}\right]  =2Y_{3}$ &  & $\left[  Y_{3},Y_{5}\right]
=2Y_{1}$\\
$\left[  Y_{1},Y_{6}\right]  =2Y_{3}$ &  & $\left[  Y_{3},Y_{6}\right]
=2Y_{1}$\\
$\left[  Y_{2},Y_{3}\right]  =-2Y_{5}$ &  & $\left[  Y_{4},Y_{5}\right]
=2Y_{1}$\\
$\left[  Y_{2},Y_{4}\right]  =-2Y_{6}$ &  & $\left[  Y_{4},Y_{6}\right]
=2Y_{2}$%
\end{tabular}
\ \ \ \label{GS968_2}%
\end{equation}
The Killing-Cartan metric for this algebra is given by:%
\begin{equation}
\mathbf{g}^{E,R,\text{red}}=\left(
\begin{array}
[c]{cccccc}%
-8 & -8 & 0 & 0 & 0 & 0\\
-8 & -16 & 0 & 0 & 0 & 0\\
0 & 0 & 8 & 8 & 0 & 0\\
0 & 0 & 8 & 16 & 0 & 0\\
0 & 0 & 0 & 0 & 8 & 8\\
0 & 0 & 0 & 0 & 8 & 16
\end{array}
\right)  \label{n3GS968}%
\end{equation}
which, in fact, posesses $2$ negative and $4$ positive eigenvalues
corresponding to $2$ compact and $4$ non-compact generators. The
expanded algebra is $\mathfrak{sl}\left(  2,\mathbb{R}\right)  \oplus
\mathfrak{sl}\left(  2,\mathbb{R}\right)  $\footnote{In this case, the identification of the algebra with an $\mathfrak{sl}(2,\mathbb{R})\oplus\mathfrak{sl}(2,\mathbb{R})$ directly  follows from inspection of the Cartan-Killing matrix (its being non-singular and having just two negative eigenvalues).
} and the two commuting $\mathfrak{sl}%
\left(  2,\mathbb{R}\right)  $ subalgebras are generated respectively by%
\begin{align}
\mathfrak{sl}\left(  2,\mathbb{R}\right)  _{1} &  =\left\{  Y_{2}-Y_{1}%
,Y_{3}-Y_{4},Y_{5}-Y_{6}\right\}  \label{n4GS968}\\
\mathfrak{sl}\left(  2,\mathbb{R}\right)  _{2} &  =\left\{  Y_{1},Y_{3}%
,Y_{5}\right\}  \nonumber
\end{align}
Furthermore, from (\ref{not05}, \ref{n2GS968}, \ref{n3GS968}) it can be directly
obtained that $\alpha=1/2$ and relations similar to (\ref{n5GS770},
\ref{n6GS770}) hold for the Cartan decomposition of this algebra.

Finally, the real compact form of $\mathcal{G}_{S,R}^{\text{red}}$ (the
complex form of $\mathcal{G}_{0,S,R}^{\text{red}}$) is given by
\begin{align}
\mathcal{G}_{k,S,R}^{\text{red}}  &  =\left(  S_{0}^{\text{red}}%
\otimes\mathcal{T}_{0}\right)  +\left(  S_{1}^{\text{red}}\otimes
i\mathcal{P}_{0}\right) \label{GS968_3}\\
&  =\left\{  \lambda_{2}\otimes i\sigma_{2},\ \lambda_{3}\otimes i\sigma
_{2}\right\}  +\left\{  \lambda_{4}\otimes i\sigma_{1},\ \lambda_{5}\otimes
i\sigma_{1},\ \lambda_{4}\otimes i\sigma_{3},\ \lambda_{5}\otimes i\sigma
_{3}\right\} \nonumber
\end{align}
It is also easy to check the relations (\ref{f7})%
\begin{align}
\mathcal{T}_{0,S,R}^{\text{red}}  &  =\mathcal{G}_{0,S,R}^{\text{red}}%
\cap\mathcal{G}_{k,S,R}^{\text{red}}=\left\{  \lambda_{2}\otimes\left(
-i\sigma_{2}\right)  ,\ \lambda_{3}\otimes\left(  -i\sigma_{2}\right)
\right\} \label{GS968_4}\\
\mathcal{P}_{0,S,R}^{\text{red}}  &  =\mathcal{G}_{0,S,R}^{\text{red}}%
\cap\left(  i\mathcal{G}_{k,S,R}^{\text{red}}\right)  =\left\{  \lambda
_{4}\otimes\sigma_{1},\ \lambda_{5}\otimes\sigma_{1},\ \lambda_{4}%
\otimes\sigma_{3},\ \lambda_{5}\otimes\sigma_{3}\right\}  \text{.}\nonumber
\end{align}

\subsection{The semigroup $S_{\left(  5\right)  }^{990}$}

For this semigroup the matrices $\mathbf{g}^{\text{red}}\left(  S_{0}\right)
$, $\mathbf{g}^{\text{red}}\left(  S_{1}\right)  $ are given by:%
\begin{equation}
\mathbf{g}^{\text{red}}\left(  S_{0}\right)  =\mathbf{g}^{\text{red}}\left(
S_{1}\right)  =\left(
\begin{array}
[c]{cc}%
4 & 0\\
0 & 4
\end{array}
\right)  \text{.} \label{n1GS990}%
\end{equation}
Again, since their eigenvalues are all positive we conclude that
semisimplicity and compactness properties will be preserved. Taking
$\mathcal{G}=\mathfrak{sl}\left(  2,\mathbb{R}\right)  $ we can predict, using
equation (\ref{res_red_final}), that the expanded algebra will have $2$
compact generators and $4$ non-compact generators and it will be then
$\mathfrak{sl}\left(  2,\mathbb{R}\right)  \oplus\mathfrak{sl}\left(
2,\mathbb{R}\right)  $. Let us see this explicitly.

Expanding with the semigroup $S_{\left(  5\right)  }^{968}$ we have,%

\begin{align}
\mathcal{G}_{0,S,R}^{\text{red}} &  =\left(  S_{0}^{\text{red}}\otimes
\mathcal{T}_{0}\right)  +\left(  S_{1}^{\text{red}}\otimes\mathcal{P}%
_{0}\right)  \label{GS990}\\
&  =\left\{  X_{\left(  1,2\right)  },X_{\left(  1,5\right)  }\right\}
+\left\{  X_{\left(  2,3\right)  },X_{\left(  2,4\right)  },X_{\left(
3,3\right)  },X_{\left(  3,4\right)  },\right\}  \nonumber\\
&  =\left\{  \lambda_{2}\otimes\left(  -i\sigma_{2}\right)  ,\ \lambda
_{5}\otimes\left(  -i\sigma_{2}\right)  \right\}  +\left\{  \lambda_{3}%
\otimes\sigma_{1},\ \lambda_{4}\otimes\sigma_{1},\ \lambda_{3}\otimes
\sigma_{3},\ \lambda_{4}\otimes\sigma_{3}\right\}  \nonumber
\end{align}
Renaming the generators as%
\begin{equation}%
\begin{tabular}
[c]{lll}%
$Y_{1}=X_{\left(  1,2\right)  }$ &  & $Y_{4}=X_{\left(  2,4\right)  }$\\
$Y_{2}=X_{\left(  1,5\right)  }$ &  & $Y_{5}=X_{\left(  3,3\right)  }$\\
$Y_{3}=X_{\left(  2,3\right)  }$ &  & $Y_{6}=X_{\left(  3,4\right)  }$.
\end{tabular}
\ \ \label{GS990_1}%
\end{equation}
and using (\ref{z3}), the commutation relations of $\mathfrak{sl}%
(2,\mathbb{R})$ and the multiplication table of $S_{\left(  5\right)  }^{990}$
(given in (\ref{y3})) we obtain:
\begin{equation}%
\begin{tabular}
[c]{lll}%
$\left[  Y_{1},Y_{3}\right]  =-2Y_{5}$ &  & $\left[  Y_{2},Y_{5}\right]
=2Y_{4}$\\
$\left[  Y_{1},Y_{4}\right]  =-2Y_{6}$ &  & $\left[  Y_{2},Y_{6}\right]
=2Y_{3}$\\
$\left[  Y_{1},Y_{5}\right]  =2Y_{3}$ &  & $\left[  Y_{3},Y_{5}\right]
=2Y_{1}$\\
$\left[  Y_{1},Y_{6}\right]  =2Y_{4}$ &  & $\left[  Y_{3},Y_{6}\right]
=2Y_{2}$\\
$\left[  Y_{2},Y_{3}\right]  =-2Y_{6}$ &  & $\left[  Y_{4},Y_{5}\right]
=2Y_{2}$\\
$\left[  Y_{2},Y_{4}\right]  =-2Y_{5}$ &  & $\left[  Y_{4},Y_{6}\right]
=2Y_{1}$%
\end{tabular}
\ \ \label{GS990_2}%
\end{equation}
whose Killing-Cartan metric is given by:%
\begin{equation}
\mathbf{g}^{E,R,\text{red}}=diag\left(  -16,-16,16,16,16,16\right)
\label{n2GS990}%
\end{equation}
which again has $2$ negative and $4$ positive eigenvalues. Then the expanded
algebra is $\mathfrak{sl}\left(  2,\mathbb{R}\right)  \oplus\mathfrak{sl}%
\left(  2,\mathbb{R}\right)  $, the $2$ commuting $\mathfrak{sl}\left(
2,\mathbb{R}\right)  $ algebras being generated respectively by%
\begin{align}
\mathfrak{sl}\left(  2,\mathbb{R}\right)  _{1} &  =\left\{  \frac{1}{2}\left(
Y_{1}-Y_{2}\right)  ,\frac{1}{2}\left(  Y_{3}-Y_{4}\right)  ,\frac{1}%
{2}\left(  Y_{5}-Y_{6}\right)  \right\}  \label{n3GS990}\\
\mathfrak{sl}\left(  2,\mathbb{R}\right)  _{2} &  =\left\{  \frac{1}{2}\left(
Y_{1}+Y_{2}\right)  ,\frac{1}{2}\left(  Y_{3}+Y_{4}\right)  ,\frac{1}%
{2}\left(  Y_{5}+Y_{6}\right)  \right\}  \nonumber
\end{align}
Besides, from (\ref{res_red_final}, \ref{not05}, \ref{n2GS990}) it is also
obtained $\alpha=1/2$ and similar relations to (\ref{n5GS770}, \ref{n6GS770})
holds for the Cartan decomposition of this algebra.

Finally, the real compact form of $\mathcal{G}_{S,R}^{\text{red}}$ is given by%
\begin{align}
\mathcal{G}_{k,S,R}^{\text{red}} &  =\left(  S_{0}^{\text{red}}\otimes
\mathcal{T}_{0}\right)  +\left(  S_{1}^{\text{red}}\otimes i\mathcal{P}%
_{0}\right)  \label{GS990_3}\\
&  =\left\{  \lambda_{2}\otimes i\sigma_{2},\ \lambda_{5}\otimes i\sigma
_{2}\right\}  +\left\{  \lambda_{3}\otimes i\sigma_{1},\ \lambda_{4}\otimes
i\sigma_{1},\ \lambda_{3}\otimes i\sigma_{3},\ \lambda_{4}\otimes i\sigma
_{3}\right\}  \text{,}\nonumber
\end{align}
so we can check again the relations (\ref{f7})%
\begin{align}
\mathcal{T}_{0,S,R}^{\text{red}} &  =\mathcal{G}_{0,S,R}^{\text{red}}%
\cap\mathcal{G}_{k,S,R}^{\text{red}}=\left\{  \lambda_{2}\otimes\left(
-i\sigma_{2}\right)  ,\ \lambda_{5}\otimes\left(  -i\sigma_{2}\right)
\right\}  \label{GS990_4}\\
\mathcal{P}_{0,S,R}^{\text{red}} &  =\mathcal{G}_{0,S,R}^{\text{red}}%
\cap\left(  i\mathcal{G}_{k,S,R}^{\text{red}}\right)  =\left\{  \lambda
_{3}\otimes\sigma_{1},\ \lambda_{4}\otimes\sigma_{1},\ \lambda_{3}%
\otimes\sigma_{3},\ \lambda_{4}\otimes\sigma_{3}\right\}  \text{.}\nonumber
\end{align}

\subsection{The semigroup $S_{\left(  5\right)  }^{991}$}

This case is different from the ones above. Indeed, the matrices
$\mathbf{g}^{\text{red}}\left(  S_{0}\right)  $, $\mathbf{g}^{\text{red}%
}\left(  S_{1}\right)  $ are given by:%
\begin{equation}
\mathbf{g}^{\text{red}}\left(  S_{0}\right)  =\left(
\begin{array}
[c]{cc}%
4 & 0\\
0 & 4
\end{array}
\right)  \ ;\ \ \mathbf{g}^{\text{red}}\left(  S_{1}\right)  =\left(
\begin{array}
[c]{cc}%
0 & 4\\
4 & 0
\end{array}
\right)  \label{n1GS991}%
\end{equation}
so $\mathbf{g}^{\text{red}}\left(  S_{0}\right)  $ has again positive
eigenvalues, while $\mathbf{g}^{\text{red}}\left(  S_{1}\right)  $ have a
positive and negative one. Then in this case the structure of the expanded
algebra is different from the cases above and indeed, taking $\mathcal{G}%
=\mathfrak{sl}\left(  2,\mathbb{R}\right)  $, the Killing-Cartan form of the
expanded algebra is now:%
\begin{equation}
\mathbf{g}^{E,R,\text{red}}=\alpha\left(
\begin{array}
[c]{cc}%
\left(
\begin{array}
[c]{cc}%
4 & 0\\
0 & 4
\end{array}
\right)  \otimes\left(  -8\right)  & 0\\
0 & \left(
\begin{array}
[c]{cc}%
0 & 4\\
4 & 0
\end{array}
\right)  \otimes\left(
\begin{array}
[c]{cc}%
8 & 0\\
0 & 8
\end{array}
\right)
\end{array}
\right)  \text{,} \label{n2GS991}%
\end{equation}
which has $4$ negative and $2$ positive eigenvalues (taking in account that
$\alpha$ is a positive number) so while semisimplicity will again be preserved
compactness properties will not. Let us see this explicitly.

Expanding explicitly with the semigroup $S_{\left(  5\right)  }^{968}$ we
have,%
\begin{align}
\mathcal{G}_{0,S,R}^{\text{red}} &  =\left(  S_{0}^{\text{red}}\otimes
\mathcal{T}_{0}\right)  +\left(  S_{1}^{\text{red}}\otimes\mathcal{P}%
_{0}\right)  \label{GS991}\\
&  =\left\{  X_{\left(  1,2\right)  },X_{\left(  1,3\right)  }\right\}
+\left\{  X_{\left(  2,4\right)  },X_{\left(  2,5\right)  },X_{\left(
3,4\right)  },X_{\left(  3,5\right)  },\right\}  \nonumber\\
&  =\left\{  \lambda_{2}\otimes\left(  -i\sigma_{2}\right)  ,\ \lambda
_{3}\otimes\left(  -i\sigma_{2}\right)  \right\}  +\left\{  \lambda_{4}%
\otimes\sigma_{1},\ \lambda_{5}\otimes\sigma_{1},\ \lambda_{4}\otimes
\sigma_{3},\ \lambda_{5}\otimes\sigma_{3}\right\}  \nonumber
\end{align}
Again the structure is very similar to the algebra (\ref{GS770}) and
(\ref{GS968}), but obviously the algebra is completly different since the
multiplication rules of the corresponding semigroups are different. Let us
define%
\begin{equation}%
\begin{tabular}
[c]{lll}%
$Y_{1}=X_{\left(  1,2\right)  }$ &  & $Y_{4}=X_{\left(  2,5\right)  }$\\
$Y_{2}=X_{\left(  1,3\right)  }$ &  & $Y_{5}=X_{\left(  3,4\right)  }$\\
$Y_{3}=X_{\left(  2,4\right)  }$ &  & $Y_{6}=X_{\left(  3,5\right)  }$.
\end{tabular}
\ \ \ \ \ \label{GS991_1}%
\end{equation}
Using the equation (\ref{z3}), the commutation relations of $\mathfrak{sl}%
(2,\mathbb{R})$ and the multiplication table of $S_{\left(  5\right)  }^{991}$
(given in (\ref{y3})) we obtain the following algebra:
\begin{equation}%
\begin{tabular}
[c]{lll}%
$\left[  Y_{1},Y_{3}\right]  =-2Y_{5}$ &  & $\left[  Y_{2},Y_{5}\right]
=2Y_{4}$\\
$\left[  Y_{1},Y_{4}\right]  =-2Y_{6}$ &  & $\left[  Y_{2},Y_{6}\right]
=2Y_{3}$\\
$\left[  Y_{1},Y_{5}\right]  =2Y_{3}$ &  & $\left[  Y_{3},Y_{5}\right]
=2Y_{2}$\\
$\left[  Y_{1},Y_{6}\right]  =2Y_{4}$ &  & $\left[  Y_{3},Y_{6}\right]
=2Y_{1}$\\
$\left[  Y_{2},Y_{3}\right]  =-2Y_{6}$ &  & $\left[  Y_{4},Y_{5}\right]
=2Y_{1}$\\
$\left[  Y_{2},Y_{4}\right]  =-2Y_{5}$ &  & $\left[  Y_{4},Y_{6}\right]
=2Y_{2}$%
\end{tabular}
\ \ \ \ \ \label{GS991_2}%
\end{equation}
whose Killing-Cartan metric is given by%
\begin{equation}
\mathbf{g}^{E,R,\text{red}}=\left(
\begin{array}
[c]{cccccc}%
-16 & 0 & 0 & 0 & 0 & 0\\
0 & -16 & 0 & 0 & 0 & 0\\
0 & 0 & 0 & 16 & 0 & 0\\
0 & 0 & 16 & 0 & 0 & 0\\
0 & 0 & 0 & 0 & 0 & 16\\
0 & 0 & 0 & 0 & 16 & 0
\end{array}
\right)  \label{n3GS991}%
\end{equation}
with eigenvalues $\left(  -16,-16,-16,-16,16,16\right)  $ corresponding to $4$
compact generators $\left\{  Y_{1},Y_{2},Y_{3}-Y_{4},Y_{5}-Y_{6}\right\}  $
and $2$ non-compact ones $\left\{  Y_{3}+Y_{4},Y_{5}+Y_{6}\right\}  $. The
expanded algebra is in this case $\mathfrak{su}\left(  2,\mathbb{R}\right)
\oplus\mathfrak{sl}\left(  2,\mathbb{R}\right)  $, the 2 commuting algebras
being generated respectively by:%
\begin{align}
\mathfrak{su}\left(  2,\mathbb{R}\right)   &  =\left\{  \frac{1}{2}\left(
Y_{1}-Y_{2}\right)  ,\frac{1}{2}\left(  Y_{3}-Y_{4}\right)  ,\frac{1}%
{2}\left(  Y_{5}-Y_{6}\right)  \right\}  \label{n4GS991}\\
\mathfrak{sl}\left(  2,\mathbb{R}\right)   &  =\left\{  \frac{1}{2}\left(
Y_{1}+Y_{2}\right)  ,\frac{1}{2}\left(  Y_{3}+Y_{4}\right)  ,\frac{1}%
{2}\left(  Y_{5}+Y_{6}\right)  \right\}  \nonumber
\end{align}
Besides, from (\ref{not05}, \ref{n2GS991}, \ref{n3GS991}) it is obtained again
$\alpha=1/2$.

\section{Summary}

\label{resum}

It was found that, under the action of the $S$-expansion procedure, some
properties of the Lie algebras (like commutativity, solvability and
nilpotency) are preserved while others (like semi-simplicity and compactness)
are not in general. This was analyzed for all kinds of expansions that can
be performed: for the direct product $\mathcal{G}_{S}=S\otimes\mathcal{G}$,
for the resonant subalgebra $\mathcal{G}_{S,R}$ and for the reduced
algebras\footnote{Remember that the existence of the resonant subalgebra and
of the reduced algebra are independent facts. As shown in picture
\ref{fig:fig0}, there are semigroups with resonant decomposition and with no
zero element and vice versa. Also it can happen that both exist
simultaneously.} $\mathcal{G}_{S,\text{red}}$ and $\mathcal{G}_{S,R,\text{red}%
}$. The same analysis was done for the expansion of a general Lie algebra
$\mathcal{G}$ on its Levi-Malcev decomposition, $\mathcal{G}=N\uplus S$. These
results are summarized in the following table:%
\begin{align*}
&  \text{\textbf{Properties preserved by the action of the} }%
S\text{\textbf{-expansion process}}\\
&
\begin{tabular}
[c]{|l|l|l|l|}\hline
\textbf{Original }$\mathcal{G}$ & \textbf{Expanded}$\mathcal{G}_{S}$ &
\textbf{Resonant }$\mathcal{G}_{S,R}$ & \textbf{Reduced }$\mathcal{G}%
_{S,R}^{\text{red}}$\\\hline
abelian & abelian & abelian & abelian\\\hline
solvable & solvable & solvable & solvable\\\hline
nilpotent & nilpotent & nilpotent & nilpotent\\\hline
compact & $\text{arbitrary}$ & $\text{arbitrary}$ & $\text{arbitrary}$\\\hline
$%
\begin{array}
[c]{c}%
\text{semisimple}\\
\mathcal{G=}S
\end{array}
$ & $%
\begin{array}
[c]{c}%
\text{arbitrary}\\
\mathcal{G}_{S}\mathcal{=}N_{\text{exp}}\uplus S_{\text{exp}}%
\end{array}
$ & $%
\begin{array}
[c]{c}%
\text{arbitrary}\\
\mathcal{G}_{S,R}\mathcal{=}N_{\text{exp,}R}\uplus S_{\text{exp,}R}%
\end{array}
$ & $%
\begin{array}
[c]{c}%
\text{arbitrary}\\
\mathcal{G}_{S,R}^{\text{red}}\mathcal{=}N_{\text{exp,}R}^{\text{red}}\uplus
S_{\text{exp,}R}^{\text{red}}%
\end{array}
$\\\hline
$%
\begin{array}
[c]{c}%
\text{arbitrary}\\
\mathcal{G=}N\uplus S
\end{array}
$ & $%
\begin{array}
[c]{c}%
\text{arbitrary}\\
\mathcal{G}_{S}\mathcal{=}N_{\text{exp}}\uplus S_{\text{exp}}%
\end{array}
$ & $%
\begin{array}
[c]{c}%
\text{arbitrary}\\
\mathcal{G}_{S,R}\mathcal{=}N_{\text{exp,}R}\uplus S_{\text{exp,}R}%
\end{array}
$ & $%
\begin{array}
[c]{c}%
\text{arbitrary}\\
\mathcal{G}_{S,R}^{\text{red}}\mathcal{=}N_{\text{exp,}R}^{\text{red}}\uplus
S_{\text{exp,}R}^{\text{red}}%
\end{array}
$\\\hline
\end{tabular}
\end{align*}

Figure \ref{fig:fig1} illustrates the general scheme of the classification
theory of Lie algebras and the arrows represent the action of the expansion
method on this classification. \textit{Grey arrows} are used when the
expansion method maps algebras of one set on to the same set, i.e., they
preserve some specific property. On the other hand, \textit{black arrows}
indicate that the expansion methods can map algebras of one specific set on to
the same set and also can lead us outside the set, to an algebra that will
have a Levi-Malcev decomposition.

\begin{figure}[th]
\centering
\includegraphics[scale=0.40]{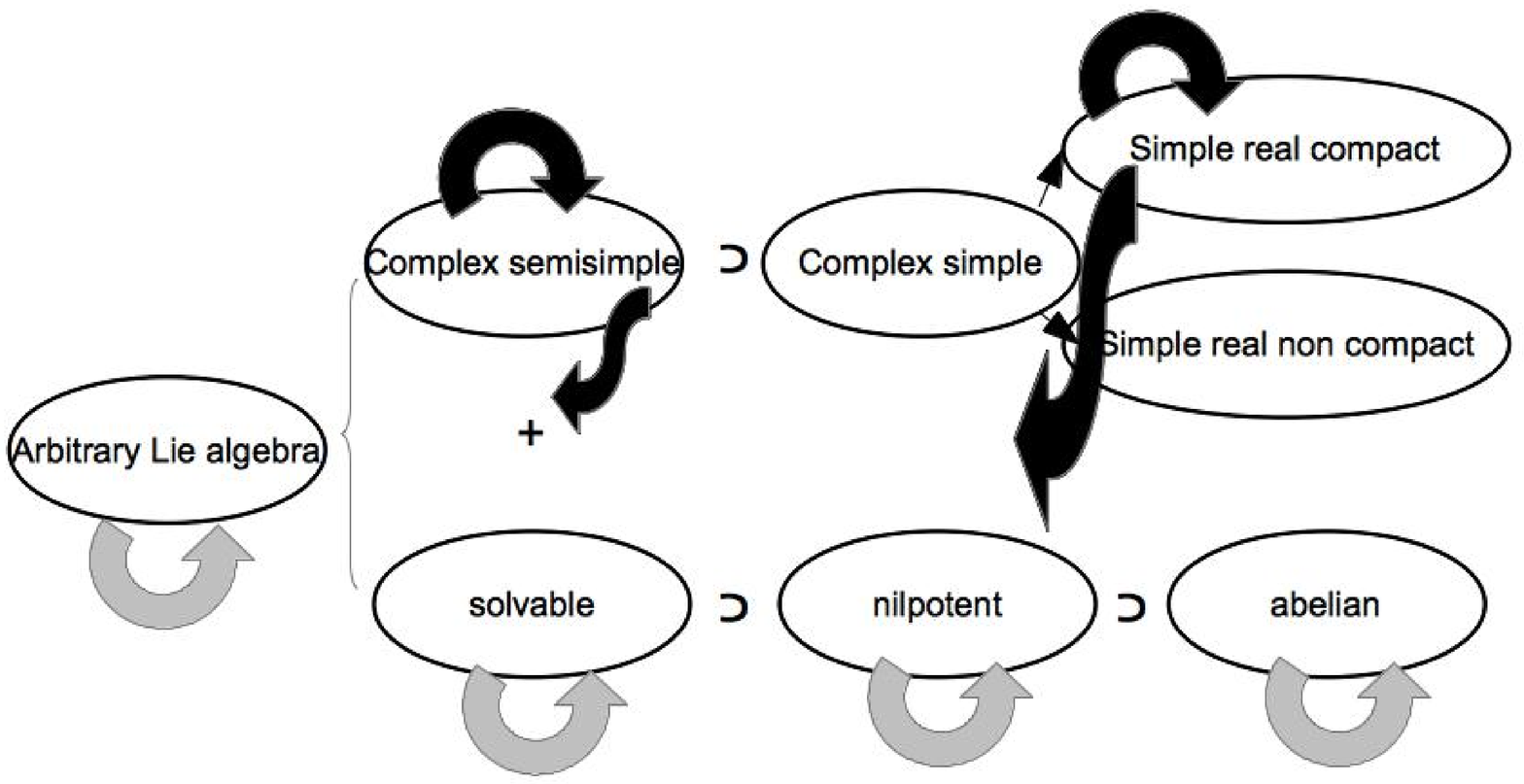}\caption{Action of the
expansion methods on the scheme of the Theory of the Classification of Lie
Algebras.}%
\label{fig:fig1}%
\end{figure}

The Cartan decomposition of the expanded algebra was also obtained when the
expansion preserves compactness. To check all these theoretical results, an
example is provided by studying all the possible expansions of the semisimple
algebra $\mathfrak{sl}(2,\mathbb{R})$ with semigroups of order up to $6$. Finally it was shown that a
classification of semigroups in terms of the eigenvalues of its associated
matrices $\mathbf{g}^{S}$, $\mathbf{g}^{S}\left(  S_{0}\right)  $,
$\mathbf{g}^{S}\left(  S_{1}\right)  $, $\mathbf{g}^{\text{red}}\left(
S_{0}\right)  $, $\mathbf{g}^{\text{red}}\left(  S_{0}\right)  $ is a useful
tool to predict semisimplicity and compactness properties of the expanded algebras.

Therefore, a procedure to answer the question: \textit{given two
Lie algebras, can they be related by some expansion procedure?} have mainly two steps:

$a)$ First we apply the theorems we have found in section \ref{preserved}. They are necessary conditions for the existence of the relation. If
the relation is forbidden by some of those therorems, then the recipe finishes here.

$b)$ If the relation is not forbidden, then the relation can still exist or
not. Then we have to study if one of the semigroups existing on each order
$n=2,...,6$ (those mentioned in section \ref{history}) with an appropiate
\textit{resonant decomposition} and maybe with a \textit{zero element} can
connect those algebras.

In this paper we have mainly developed the step $a)$. On the other hand, part
$b)$ can be performed by hand (as for example in \cite{Bianchi}) but in more
complicated cases it is necessary to use computer programs. A report with a
set of computer programs (some of them used here to study expansions of
$\mathfrak{sl}(2,\mathbb{R})$), more specifically a Java library, to study all
the possible relations by means of an S-expansion between two arbitrary Lie
algebras is a work in progress \cite{Nadal}. This library will give
\textit{resonant subalgebras} and \textit{reduced algebras} for expansions
performed with any semigroup and will also contain algorithms to classify the
expansions through comparing, by means of isomorphisms or anti-isomorphisms,
with those semigroups enumerated and characterized in section \ref{history}.

\section{Comments and Possible Developments}

\label{Com}

The results used in this paper can be extended to study what happens to the Gauss and Iwasawa decompositions under the expansion
procedure. Those decompositions play a crucial role in the theory of
representation of Lie algebras so this study could be very useful for
obtaining a general picture of the possible relations between arbitrary
algebras, mainly in those used in physics.

As shown in the table of section 4.2 the fraction of semigroups preserving semisimplicity is very small (on each level of the expansion procedure) and expansions performed with those semisgroups are in general they are direct sums of the original one. Therefore those expansions can be not so much interesing from the point of view of obtaining  fundamentally new objects. However, there is an interesting observation that can be made here:

\textit{For expansions with semigroup of order }$n=3,4,5$\textit{ if
semisimplicity is preserved, then the expanded algebra always contains the
original algebra as a subalgebra. }

This made us conjecture that there is no expansion procedure (at least an
expansion procedure that just uses semigroups with the conditions considered here) that permits to relate the
simple Lie algebras (the classical Lie algebras $A_{n}$, $B_{n}$, $C_{n}$,
$D_{n}$ and the special ones) among each other. A study to prove this conjecture
and some new physical applications with these results (particularly in gauge
Chern-Simons theories of gravity) is a work in progress (see \cite{Nadal}).

Appart from this fact, for physical applications it seems to be more interesting to use those semigroups that broke semisimplicity in the sense that they permit us to obtain fundamentally new objects. This is the case of two examples already mentioned in the introduction: $a)$ the non semisimple M super algebra, which is obtained as an expansion of the semisimple osp(32/1) superalgebra and $b)$ the non semisimple $\mathfrak{B}$ algebra obtained as an expansion of the AdS algebra. In both cases the original algebra is not a subalgebra of the expanded one, and their structure is richer than the expanded algebras obtained with semigroups preserving semisimplicity. In ref. \cite{Inos} propose a method to obtain all the semigroups, that in the same way for the $\mathfrak{B}$ algebra, permit us to obtain standard general relativity as a special limit of a CS theory in five dimensions.

On the other hand, as mentioned in \cite{Nesterenko}, it would be interesting
to know whether $S$-expansions fit to the classification of solvable Lie
algebras of a fixed dimension by means of $S$-expansions of simple
(semisimple) Lie algebras of the same dimension. The theoretical results proposed in this work could be useful to solve that problem.

\section*{Acknowledgements}

L. A. \& M.T. were supported by the Italian MIUR-PRIN contract 2009KHZKRX-007 Symmetries of the Universe and of the Fundamental Interactions. N. M. \& F.N. wish to thank
L. Andrianopoli, R. D'Auria and M. Trigiante for their kind hospitality at
Dipartimento di Scienza Applicata e Tecnologia (DISAT) of Politecnico di
Torino, where part of this work was done. N.M. was supported by FONDECYT
(Chile) grant 3130445; he is also grateful to O. Mi\v{s}kovi\v{c} and P. Salgado
for many valuable discussions. F. N. wants to thank CSIC for a JAE-Predoc
grant cofunded by the European Social Fund.

\section*{A. Reducible representations}

Let us consider a semi-simple Lie algebra $\mathcal{G}$ with a decomposition
$\mathcal{G=H+K}$ where $\mathcal{H}$ is a subalgebra. Then $\mathcal{G}%
/\mathcal{H}$ is a representation of $\mathcal{H}$. If besides $\mathcal{H}$
is semisimple, then $\mathcal{G}$ is a completly reducible representation with
respect to the adjoint action of $\mathcal{H}$ and therefore%
\begin{equation}
\left[  \mathcal{H},\mathcal{K}\right]  \subset\mathcal{K} \label{prorpe}%
\end{equation}
If we use the indices $i,j$ for $\mathcal{H}$ and $a,b\,\ $for $\mathcal{K}$
then we have that the Killing-Cartan metric of $\mathcal{G}$ restricted to $\mathcal{H}$ is
given by%
\begin{equation}
g_{ij}^{\mathcal{G}}=C_{ik}^{l}C_{jl}^{k}+C_{ia}^{b}C_{jb}^{a}\text{,}
\label{prope2}%
\end{equation}
where (\ref{prorpe}) was used, while the Killing-Cartan metric for $\mathcal{H}$ is given by%
\begin{equation}
g_{ij}^{\mathcal{H}}=C_{ik}^{l}C_{jl}^{k}\text{.} \label{prope3}%
\end{equation}
Considering the adjoint action in $\mathcal{K}$ of a generator in
$\mathcal{H}$, $T_{i}^{\mathcal{K}}=\left(  C_{ia}^{b}\right)  $,we have that%
\[
C_{ia}^{b}C_{jb}^{a}=Tr\left(  T_{i}^{\mathcal{K}}T_{j}^{\mathcal{K}}\right)
=\alpha^{\mathcal{K}}g_{ij}^{\mathcal{H}}%
\]
where the last equality comes from the fact that for whatever matrix
representation of the generators of $\mathcal{H}$, $T_{i}^{\mathcal{R}}$, the
symmetric invariant form constructed as $Tr\left(  T_{i}^{\mathcal{R}}%
T_{j}^{\mathcal{R}}\right)  $ must be proportional to $g_{ij}^{\mathcal{H}}$.
The reason is that $g_{ij}^{\mathcal{H}}$ is the only symmetric 2-index invariant form
that exist for a semisimple Lie algebra $\mathcal{H}$. The constant
$\alpha^{\mathcal{K}}$ depends on the dimensions $d_{\mathcal{H}}$,
$d_{\mathcal{K}}$ of $\mathcal{H}$ and of the representation $\mathcal{K}$ and
is given by%
\begin{equation}
\alpha^{\mathcal{K}}=\frac{d_{\mathcal{K}}c_{\mathcal{K}}}{d_{\mathcal{H}}}
\label{prope5}%
\end{equation}
where $c_{\mathcal{K}}$ is the positive eigenvalue of the second Casimir
operator, relative to the representation $\mathcal{K}$ of $\mathcal{H}$ (see
Ref. \cite{L3}: Gilmore, ex. 13 pag. 276). Therefore it is direct to see
that:
\[
g_{ij}^{\mathcal{G}}=\left(  1+\frac{d_{\mathcal{K}}c_{\mathcal{K}}%
}{d_{\mathcal{H}}}\right)  g_{ij}^{\mathcal{H}}\text{.}%
\]

\section*{B. Expansion of a general Lie algebra}

Here we give the proof of theorem $4$ of section \ref{gener} by showing that
$N_{\text{exp}}=E_{N}\uplus N^{\prime}$ \ is the radical of the expanded
algebra $\mathcal{G}_{\mathcal{S}}$. Let's prove first that $N_{\text{exp}}$
is an ideal:%
\begin{align*}
\left[  N_{\text{exp}},N_{\text{exp}}\right]   &  =\left[  E_{N}\uplus
N^{\prime},E_{N}\uplus N^{\prime}\right] \\
&  \subset\left[  E_{N},E_{N}\right]  +\left[  E_{N},N^{\prime}\right]
+\left[  N^{\prime},N^{\prime}\right] \\
&  \subset E_{N}+E_{N}+N^{\prime}=N_{\text{exp}}%
\end{align*}%
\begin{align*}
\left[  N_{\text{exp}},S_{\text{exp}}\right]   &  =\left[  N_{\text{exp}%
},S_{\text{exp}}\right]  =\left[  E_{N}\uplus N^{\prime},S_{\text{exp}}\right]
\\
&  =\left[  E_{N},S_{\text{exp}}\right]  +\left[  N^{\prime},S_{\text{exp}%
}\right] \\
&  \subset E_{N}+N^{\prime}=N_{\text{exp}}%
\end{align*}
so $N_{\text{exp}}$ it is an ideal because,
\begin{align*}
\left[  N_{\text{exp}},N_{\text{exp}}\right]   &  \subset N_{\text{exp}}\\
\left[  N_{\text{exp}},S_{\text{exp}}\right]   &  \subset N_{\text{exp}%
}\text{.}%
\end{align*}
Now let us prove that $N_{\text{exp}}$ is solvable. We see that%
\[
N_{\text{exp}}=E_{N}\uplus N^{\prime}=\left\{  X_{\left(  i_{N},\alpha\right)
},X_{\left(  i_{S},\alpha\right)  _{N^{\prime}}}\right\}
\]
where $X_{\left(  i_{N},\alpha\right)  }\in E_{N}$ and $X_{\left(
i_{S},\alpha\right)  _{N^{\prime}}}\in N^{\prime}$. \ We have to study the
behavior of $N_{\text{exp}}^{\left(  n\right)  }$ defined by%
\[
N_{\text{exp}}^{\left(  n\right)  }=\left[  N_{\text{exp}}^{\left(
n-1\right)  }\text{,}N_{\text{exp}}^{\left(  n-1\right)  }\right]
\]
where the set
\[
N_{\text{exp}}^{\left(  n\right)  }=\left\{  X_{\left(  i_{N}^{\left(
n\right)  },\alpha^{\left(  n\right)  }\right)  },X_{\left(  i_{S}^{\left(
n\right)  },\alpha^{\left(  n\right)  }\right)  _{N^{\prime}}}\right\}
\]
have the following commutation relations:%
\begin{align}
\left[  X_{\left(  i_{N}^{\left(  n-1\right)  },\alpha^{\left(  n-1\right)
}\right)  },X_{\left(  j_{N}^{\left(  n-1\right)  },\beta^{\left(  n-1\right)
}\right)  }\right]   &  =C_{i_{N}^{\left(  n-1\right)  }j_{N}^{\left(
n-1\right)  }}^{k_{N}^{\left(  n\right)  }}K_{\alpha^{\left(  n-1\right)
}\beta^{\left(  n-1\right)  }}^{\gamma^{\left(  n\right)  }}X_{\left(
k_{N}^{\left(  n\right)  },\gamma^{\left(  n\right)  }\right)  }\label{aa1}\\
\left[  X_{\left(  i_{N}^{\left(  n-1\right)  },\alpha^{\left(  n-1\right)
}\right)  },X_{\left(  j_{S}^{\left(  n-1\right)  },\beta^{\left(  n-1\right)
}\right)  _{N^{\prime}}}\right]   &  =C_{i_{N}^{\left(  n-1\right)  }\left(
j_{S}\right)  _{N^{\prime}}^{\left(  n-1\right)  }}^{k_{N}^{\left(  n\right)
}}K_{\alpha^{\left(  n-1\right)  }\beta_{N^{\prime}}^{\left(  n-1\right)  }%
}^{\gamma^{\left(  n\right)  }}X_{\left(  k_{N}^{\left(  n\right)  }%
,\gamma^{\left(  n\right)  }\right)  }\label{aa2}\\
\left[  X_{\left(  i_{S}^{\left(  n-1\right)  },\alpha^{\left(  n-1\right)
}\right)  _{N^{\prime}}},X_{\left(  j_{S}^{\left(  n-1\right)  }%
,\beta^{\left(  n-1\right)  }\right)  _{N^{\prime}}}\right]   &  =C_{\left(
i_{S}\right)  _{N^{\prime}}^{\left(  n-1\right)  }\left(  j_{S}\right)
_{N^{\prime}}^{\left(  n-1\right)  }}^{\left(  k_{S}\right)  _{N^{\prime}%
}^{\left(  n\right)  }}K_{\alpha_{N^{\prime}}^{\left(  n-1\right)  }%
\beta_{N^{\prime}}^{\left(  n-1\right)  }}^{\gamma_{N^{\prime}}^{\left(
n\right)  }}X_{\left(  k_{S}^{\left(  n\right)  },\gamma^{\left(  n\right)
}\right)  _{N^{\prime}}} \label{aa3}%
\end{align}
For is $N_{\text{exp}}^{\left(  1\right)  }$ we have%
\begin{align*}
N_{\text{exp}}^{\left(  1\right)  }  &  =\left[  N_{\text{exp}}^{\ \left(
0\right)  },N_{\text{exp}}^{\left(  0\right)  }\right] \\
&  =\left[  E_{N},E_{N}\right]  +\left[  E_{N},N^{\prime}\right]  +\left[
N^{\prime},N^{\prime}\right] \\
&  \subset E_{N}^{\left(  1\right)  }+N^{\prime\left(  1\right)  }%
+E_{N}=\left\{  E_{N}^{\left(  1\right)  },N^{\prime\left(  1\right)  }%
,E_{N}^{\left(  0\right)  }\right\}
\end{align*}
where $N_{\text{exp}}^{\ \left(  0\right)  }=N_{\text{exp}}$, $E_{N}^{\left(
1\right)  }=\left[  E_{N},E_{N}\right]  $, $N^{\prime\left(  1\right)
}=\left[  N^{\prime},N^{\prime}\right]  $ and we have used the fact that
$\left[  E_{N},N^{\prime}\right]  \subset E_{N}$ as it can be seen in equation
(\ref{aa2}). Then for $N_{\text{exp}}^{\left(  2\right)  }$ we have%
\begin{align*}
N_{\text{exp}}^{\left(  2\right)  }  &  =\left[  N_{\text{exp}}^{\ \left(
1\right)  },N_{\text{exp}}^{\left(  1\right)  }\right] \\
&  =\left[  E_{N}^{\left(  1\right)  }+N^{\prime\left(  1\right)  }%
+E_{N}^{\left(  0\right)  },E_{N}^{\left(  1\right)  }+N^{\prime\left(
1\right)  }+E_{N}^{\left(  0\right)  }\right] \\
&  \subset\left\{  E_{N}^{\left(  2\right)  },N^{\prime\left(  2\right)
},E_{N}^{\left(  1\right)  },E_{N}^{\left(  0,1\right)  }\right\}
\end{align*}
where $E_{N}^{\left(  2\right)  }=\left[  E_{N}^{\left(  1\right)  }%
,E_{N}^{\left(  1\right)  }\right]  $, $N^{\prime\left(  2\right)  }=\left[
N^{\prime\left(  1\right)  },N^{\prime\left(  1\right)  }\right]  $ and
$E_{N}^{\left(  0,1\right)  }=\left[  E_{N}^{\left(  0\right)  }%
,N^{\prime\left(  1\right)  }\right]  $. Thus, as $N^{\prime}$ is solvable,
there exist som $m$ such that%
\[
N_{\text{exp}}^{\left(  m\right)  }=\left\{  E_{N}^{\left(  m\right)  }%
,E_{N}^{\left(  m-1\right)  },E_{N}^{\left(  m-2,m-1\right)  },...,E_{N}%
^{\left(  ...\right)  }\right\}  \subset E_{N}%
\]
On the other hand, as $E_{N}$ is solvable, there exist some $n$ such that
$E_{N}^{\left(  m\right)  }=0$. Therefore $N_{\text{exp}}$ is solvable too,
because there exist a value $p=m+n$ such that
\[
N_{\text{exp}}^{\left(  p\right)  }=N_{\text{exp}}^{\left(  m+n\right)
}=\left(  N_{\text{exp}}^{\left(  m\right)  }\right)  ^{\left(  n\right)
}=E_{N}^{\left(  n\right)  }=0\text{.}%
\]

Finally, suppose $N_{exp}=E_N\oplus N'$ is not maximal. This means that there exists a generator $X\in S_{exp}$ such that $N_{exp}'=N_{exp}\oplus \{X\}$ is a solvable ideal of $\mathcal{G}_{\mathcal{S}}$. Define $N''=N'\oplus \{X\}\subset E_S$. We have that
\begin{equation}
N''=N_{exp}'\cap E_S\,,
\end{equation}
so that $N''$ is a solvable Lie subalgebra of $E_S$. Moreover, if $S'=S_{exp}\ominus \{X\}$,
\begin{equation}
\mathcal{G}_{\mathcal{S}}=N_{exp}'\biguplus S'\,,
\end{equation}
so that, being $N_{exp}'$ an ideal by assumption, $[S',\,N_{exp}']\subset N_{exp}'$. Restricting both sides to $E_S$, we then have:
\begin{equation}
[S',\,N'']\subset N''\,,
\end{equation}
that is $N''=N'\oplus \{X\}$ is a solvable ideal of $E_S$, which cannot be since $N'$ is maximal in $E_S$.

\section*{C. Cartan decomposition under the S-expansion}

\textbf{The expanded semisimple algebra:}

Here is given the proof of theorem $5$ of section \ref{Cartandec}, i.e.
that (\ref{CC_1}-\ref{SCD_02}) is the Cartan decomposition of $\mathcal{G}%
_{0,S}$ when $\mathcal{G}_{k,S}$ is compact. This is done by providing a
conjugation $\sigma_{S}$ of $\mathcal{G}$ with respect to $\mathcal{G}_{0,S}$
such that:
\begin{equation}
\sigma_{S}\left(  \mathcal{G}_{k,S}\right)  \subset\mathcal{G}_{k,S}\text{
}\ \label{CD1}%
\end{equation}
and then by showing the following relations:%
\begin{align}
\mathcal{T}_{0,S} &  =\mathcal{G}_{0,S}\cap\mathcal{G}_{k,S}\label{CD2}\\
\mathcal{P}_{0,S} &  =\mathcal{G}_{0,S}\cap\left(  i\mathcal{G}_{k,S}\right)
\label{CD3}%
\end{align}

Let us find how the explicit form of $\sigma_{S}$ is found. Consider the elements
$A=a^{i}X_{i}$ $\in\mathcal{G}$ and $B=b^{\left(  i,\alpha\right)  }X_{\left(
i,\alpha\right)  }\in\mathcal{G}_{S}$ and let $\left\{  X_{i}\right\}
=\left\{  X_{i^{\left(  0\right)  }},iX_{i^{\left(  0\right)  }}\right\}
_{i^{\left(  0\right)  }=1}^{\dim\mathcal{G}_{0}}$ and $\left\{  X_{\left(
i,\alpha\right)  }\right\}  =\left\{  X_{\left(  i^{\left(  0\right)  }%
,\alpha\right)  },iX_{\left(  i^{\left(  0\right)  },\alpha\right)  }\right\}
_{i^{\left(  0\right)  }=1}^{\dim\mathcal{G}_{0}}$ be respectively the bases
of $\mathcal{G}$ and $\mathcal{G}_{S}$. Then, we can write%
\begin{align}
A &  =a^{i}X_{i}=a^{i^{\left(  0\right)  }}X_{i^{\left(  0\right)  }}%
+i\tilde{a}^{i^{\left(  0\right)  }}X_{i^{\left(  0\right)  }}\label{CD4_0}\\
B &  =b^{\left(  i,\alpha\right)  }X_{\left(  i,\alpha\right)  }=b^{\left(
i^{\left(  0\right)  },\alpha\right)  }X_{\left(  i^{\left(  0\right)
},\alpha\right)  }+i\tilde{b}^{\left(  i^{\left(  0\right)  },\alpha\right)
}X_{\left(  i^{\left(  0\right)  },\alpha\right)  }\nonumber
\end{align}
where $a^{i^{\left(  0\right)  }},\tilde{a}^{i^{\left(  0\right)  }%
},b^{\left(  i^{\left(  0\right)  },\alpha\right)  }$ and $\tilde{b}^{\left(
i^{\left(  0\right)  },\alpha\right)  }$ are real constants. Let us also
define a mapping $\sigma_{S}:\mathcal{G}_{S}\rightarrow\mathcal{G}_{S}$ such
that
\begin{align}
\sigma_{S}\left(  \alpha B_{1}+\beta B_{2}\right)    & =\bar{\alpha}\sigma
_{S}\left(  B_{1}\right)  +\bar{\beta}\sigma_{S}\left(  B_{2}\right)  \text{,
}\forall B_{1},B_{2}\in\mathcal{G}_{S}\text{ , }\forall\alpha,\beta\in%
\mathbb{C}
\label{CD4_1}\\
\sigma_{S}\left(  \lambda_{\alpha}\otimes A\right)    & =\lambda_{\alpha
}\otimes\sigma\left(  A\right)  \text{,}\ \forall\lambda_{\alpha}\in S\text{,
}\forall A\in\mathcal{G}\nonumber
\end{align}
where $\sigma$ is the conjugation of $\mathcal{G}$ with respect to
$\mathcal{G}_{0}$ and where $\bar{\alpha}$ denote the complex conjutate of
$\alpha$. Then the mapping $\sigma_{S}$ of an arbitrary element $B\in
\mathcal{G}_{S}$ can be expressed in terms of $\sigma$ as follows:%
\begin{align}
\sigma_{S}\left(  B\right)    & =\sigma_{S}\left(  b^{\left(  i,\alpha\right)
}X_{\left(  i,\alpha\right)  }\right)  \label{CD4}\\
& =\sigma_{S}\left(  b^{\left(  i^{\left(  0\right)  },\alpha\right)
}X_{\left(  i^{\left(  0\right)  },\alpha\right)  }+i\tilde{b}^{\left(
i^{\left(  0\right)  },\alpha\right)  }X_{\left(  i^{\left(  0\right)
},\alpha\right)  }\right)  \nonumber\\
& =\lambda_{\alpha}\otimes\sigma\left(  b^{\left(  i^{\left(  0\right)
},\alpha\right)  }X_{i^{\left(  0\right)  }}+i\tilde{b}^{\left(  i^{\left(
0\right)  },\alpha\right)  }X_{i^{\left(  0\right)  }}\right)  \nonumber
\end{align}
Then it is straightforward to show that $\sigma_{S}$ is a conjugation of $\mathcal{G}%
_{S}$ with respect to $\mathcal{G}_{0,S}$, i.e., that also satisfies
\[
\sigma_{S}\left[  B_{1},B_{2}\right]  =\left[  \sigma_{S}\left(  B_{1}\right)
,\sigma_{S}\left(  B_{2}\right)  \right]  \text{ \  }\forall B_{1},B_{2}\in
\mathcal{G}_{S}\text{ \ \ and \ }\left(  \sigma_{S}\right)  ^{2}=I\text{.}%
\]

Now let us prove (\ref{CD1}), i.e., that $\mathcal{G}_{k,S}$ is invariant
under the conjugation (\ref{CD4}). In fact, consider the action of $\sigma
_{S}$ on an abitrary element $K=k^{\left(  i^{\left(  k\right)  }%
,\alpha\right)  }X_{\left(  i^{\left(  k\right)  },\alpha\right)  }%
\in\mathcal{G}_{k,S}$:%
\begin{align*}
\sigma_{S}\left(  k^{\left(  i^{\left(  k\right)  },\alpha\right)  }X_{\left(
i^{\left(  k\right)  },\alpha\right)  }\right)    & =k^{\left(  i^{\left(
k\right)  },\alpha\right)  }\sigma_{S}\left(  \lambda_{\alpha}\otimes
X_{i^{\left(  k\right)  }}\right)  =k^{\left(  i^{\left(  k\right)  }%
,\alpha\right)  }\lambda_{\alpha}\otimes\sigma\left(  X_{i^{\left(  k\right)
}}\right)  \\
& =k^{\left(  i^{\left(  k\right)  },\alpha\right)  }\lambda_{\alpha}\otimes
X_{i^{\left(  k\right)  }}=k^{\left(  i^{\left(  k\right)  },\alpha\right)
}\lambda_{\alpha}\otimes X_{i^{\left(  k\right)  }}\\
& =k^{\left(  i^{\left(  k\right)  },\alpha\right)  }X_{\left(  i^{\left(
k\right)  },\alpha\right)  }\in\mathcal{G}_{k,S}%
\end{align*}
where we have used (\ref{CD4_1}) and $\sigma\left(  \mathcal{G}_{k}\right)
\subset\mathcal{G}_{k}$. In this way (\ref{CD1}) is satisfied.

Now proving (\ref{CD2}) is easy because%
\[
\mathcal{T}_{0,S}=\left\{  X_{\left(  i_{0}^{\left(  0\right)  }%
,\alpha\right)  }\right\}  =\left\{  \lambda_{\alpha}\otimes X_{i_{0}^{\left(
0\right)  }}\right\}
\]
where by hypothesis $X_{i_{0}^{\left(  0\right)  }}\in\mathcal{G}_{0}%
\cap\mathcal{G}_{k}$. Then $X_{\left(  i_{0}^{\left(  0\right)  }%
,\alpha\right)  }$ is in $\mathcal{G}_{0,S}$ and also in $\mathcal{G}_{k,S}$,
i.e., $X_{\left(  i_{0}^{\left(  0\right)  },\alpha\right)  }\in
\mathcal{G}_{0,S}\cap\mathcal{G}_{k,S}$ therefore (\ref{CD2}) is true.

On the other hand to prove (\ref{CD3}) we just have to note that%
\[
\mathcal{P}_{0,S}=\left\{  X_{\left(  i_{1}^{\left(  0\right)  }%
,\alpha\right)  }\right\}  =\left\{  \lambda_{\alpha}\otimes X_{i_{1}^{\left(
0\right)  }}\right\}
\]
where $X_{i_{1}^{\left(  0\right)  }}\in\mathcal{G}_{0}\cap\left(
i\mathcal{G}_{k}\right)  $. Then $X_{\left(  i_{1}^{\left(  0\right)  }%
,\alpha\right)  }\in\mathcal{G}_{0,S}$ and also in $i\mathcal{G}_{k,S}$, i.e.,
$X_{\left(  i_{1}^{\left(  0\right)  },\alpha\right)  }\in\mathcal{G}%
_{0,S}\cap\left(  i\mathcal{G}_{k,S}\right)  $ so we have also proven
(\ref{CD3}).

\bigskip

\textbf{The semisimple resonant subalgebra}

Here we give the proof of theorem $6$ of section \ref{Cartandec}. We have to find
a compact real form, $\mathcal{G}_{k,S,R}$, of $\mathcal{G}_{S,R}$ (the
complex form of $\mathcal{G}_{0,S,R}$) satisfying the conditions (\ref{BCD3}),
that in this case read as:%
\begin{equation}
\sigma_{S,R}\left(  \mathcal{G}_{k,S,R}\right)  \subset\mathcal{G}%
_{k,S,R}\text{ and }\ \label{CDreS4}%
\end{equation}%
\begin{align}
\mathcal{T}_{0,S,R} &  =\mathcal{G}_{0,S,R}\cap\mathcal{G}_{k,S,R}%
\label{CDreS5}\\
\mathcal{P}_{0,S,R} &  =\mathcal{G}_{0,S,R}\cap\left(  i\mathcal{G}%
_{k,S,R}\right)  \label{CDreS6}%
\end{align}
where $\sigma_{S,R}$ is a conjugation in $\mathcal{G}_{S,R}$ with respect to
$\mathcal{G}_{0,S,R}$.

As we saw before, the expansion of the compact algebra $\mathcal{G}_{k}$,%
\begin{equation}
\mathcal{G}_{k,S}=S\otimes\mathcal{G}_{k}\text{,}\label{f1}%
\end{equation}
with%
\begin{equation}
\mathcal{G}_{k}=\mathcal{T}_{0}+i\mathcal{P}_{0}\text{.}\label{f2}%
\end{equation}
is compact when $\xi_{\alpha}>0$. Besides $\mathcal{G}_{k}$ satisfies the
resonant condition, as can be seen in (\ref{mario_res0}-\ref{mario_res}), so%
\begin{equation}
\mathcal{G}_{k,S,R}=\left(  S_{0}\otimes\mathcal{T}_{0}\right)  +\left(
S_{1}\otimes i\mathcal{P}_{0}\right)  \text{,}\label{f3}%
\end{equation}
is the resonant subalgebra of $\mathcal{G}_{k,S}$ and it is compact because it
is a subalgebra of a compact Lie algebra, $\mathcal{G}_{k,S}$.

Let us prove now that (\ref{CDreS4}-\ref{CDreS6}) are satisfied. Considering
that
\[
\left\{  X_{i^{\left(  0\right)  }}\right\}  _{i^{\left(  0\right)  }=1}%
^{\dim\mathcal{G}_{0}}\text{,\ \ \ }\left\{  X_{i_{0}^{\left(  0\right)  }%
}\right\}  _{i_{0}^{\left(  0\right)  }=1}^{\dim\mathcal{T}_{0}}\text{,
\ \ }\left\{  X_{i_{1}^{\left(  0\right)  }}\right\}  _{i_{0}^{\left(
0\right)  }=1}^{\dim\mathcal{P}_{0}}%
\]
are respectively bases of $\mathcal{G}$, $\mathcal{T}_{0}$ and $\mathcal{P}%
_{0}$ we have that an arbitrary element in $\mathcal{G=T}_{0}+\mathcal{P}_{0}$
can be written as follows%
\begin{align}
A &  =a^{i}X_{i}=a^{i^{\left(  0\right)  }}X_{i^{\left(  0\right)  }}%
+i\tilde{a}^{i^{\left(  0\right)  }}X_{i^{\left(  0\right)  }}\label{f3_2}\\
&  =a^{i_{r}^{\left(  0\right)  }}X_{i_{r}^{\left(  0\right)  }}+i\tilde
{a}^{i_{r}^{\left(  0\right)  }}X_{i_{r}^{\left(  0\right)  }}\text{, with
}r=0,1\text{,}\nonumber
\end{align}
where a sum on $r=0,1$ is also assumed. Then an arbitrary element on
$\mathcal{G}_{S,R}$ can be written as%
\[
B=%
{\displaystyle\sum\limits_{r=0,1}}
\left(  b^{\left(  i_{r}^{\left(  0\right)  },\alpha_{r}\right)  }X_{\left(
i_{r}^{\left(  0\right)  },\alpha_{r}\right)  }+i\tilde{b}^{\left(
i_{r}^{\left(  0\right)  },\alpha_{r}\right)  }X_{\left(  i_{r}^{\left(
0\right)  },\alpha_{r}\right)  }\right)
\]
so the conjugation $\sigma_{S}$, defined before, acting on this element gives%
\begin{equation}
\sigma_{S}\left(  B\right)  =%
{\displaystyle\sum\limits_{r=0,1}}
\lambda_{\alpha_{r}}\otimes\sigma\left(  b^{\left(  i_{r}^{\left(  0\right)
},\alpha_{r}\right)  }X_{i_{r}^{\left(  0\right)  }}+i\tilde{b}^{\left(
i_{r}^{\left(  0\right)  },\alpha_{r}\right)  }X_{i_{r}^{\left(  0\right)  }%
}\right)  \text{ ,\ \ }\forall B\in\mathcal{G}_{S,R}\label{f4__}%
\end{equation}
for $r=0,1$. Now let's consider an abitrary element $K\in\mathcal{G}%
_{k,S,R}\,$, i.e.,
\[
K=k^{\left(  i_{0}^{\left(  k\right)  },\alpha_{0}\right)  }X_{\left(
i_{0}^{\left(  k\right)  },\alpha_{0}\right)  }+\tilde{k}^{\left(  i_{\bar{0}%
}^{\left(  k\right)  },\alpha_{1}\right)  }X_{\left(  i_{\bar{0}}^{\left(
k\right)  },\alpha_{1}\right)  }%
\]
where $i_{0}^{\left(  k\right)  }$ and $i_{\bar{0}}^{\left(  k\right)  }$ are
indices living on $\mathcal{G}_{k,\mathcal{T}_{0}}=\mathcal{G}_{k}%
\cap\mathcal{T}_{0}$ and on $\mathcal{G}_{k}/\mathcal{G}_{k,\mathcal{T}_{0}}$
respectively. Then,%

\begin{align*}
\sigma_{S}\left(  K\right)   &  =k^{\left(  i_{0}^{\left(  k\right)  }%
,\alpha_{0}\right)  }\sigma_{S}\left(  X_{\left(  i_{0}^{\left(  k\right)
},\alpha_{0}\right)  }\right)  +\tilde{k}^{\left(  i_{\bar{0}}^{\left(
k\right)  },\alpha_{1}\right)  }\sigma_{S}\left(  X_{\left(  i_{\bar{0}%
}^{\left(  k\right)  },\alpha_{1}\right)  }\right)  \\
&  =k^{\left(  i_{0}^{\left(  k\right)  },\alpha_{0}\right)  }\lambda
_{\alpha_{0}}\otimes\sigma_{S}\left(  X_{i_{0}^{\left(  k\right)  }}\right)
+\tilde{k}^{\left(  i_{\bar{0}}^{\left(  k\right)  },\alpha_{1}\right)
}\lambda_{\alpha_{1}}\otimes\sigma_{S}\left(  X_{i_{\bar{0}}^{\left(
k\right)  }}\right)
\end{align*}
and as
\begin{align*}
\lambda_{\alpha_{0}}\otimes\sigma_{S}\left(  X_{i_{0}^{\left(  k\right)  }%
}\right)   &  \subset S_{0}\otimes\mathcal{T}_{0}\\
\lambda_{\alpha_{1}}\otimes\sigma_{S}\left(  X_{i_{\bar{0}}^{\left(  k\right)
}}\right)   &  \subset S_{1}\otimes\left(  i\mathcal{P}_{0}\right)
\end{align*}
we have that%
\[
\sigma_{S}\left(  K\right)  \subset\left(  S_{0}\otimes\mathcal{T}_{0}\right)
+\left(  S_{1}\otimes\left(  i\mathcal{P}_{0}\right)  \right)  =\mathcal{G}%
_{k,S,R}%
\]
so (\ref{CDreS4}) is proved to be true. In the same way is it possible to show
(\ref{CDreS5}) and (\ref{CDreS6}).


\begin{thebibliography}{99}                                                                                               %

\bibitem{Segal} I.E. Segal, \textit{A class of operator algebras which are
determined by groups}, Duke Math. J. \textbf{18}, 221-265 (1951)

\bibitem{IW} E. In\"{o}n\"{u} and E.P. Wigner, \textit{On the contraction of
groups and their representations}, Proc. Nat. Acad. Sci. U.S.A. \textbf{39},
510-524 (1953); 
E. In\"{o}n\"{u}, \textit{Contractions of Lie groups and their representations}. In: G\"{u}rsey, F. (ed.) \textit{Group Theoretical Concepts in Elementary
particle Physics} pp 391-402.  Gordon and Breach, New York (1964)

\bibitem {WW}E. Weimar-Woods, \textit{Contractions of Lie algebras:
generalized Inonu-Wigner contractions versus graded contractions}, J. Math.
Phys. \textbf{36}, 4519-4548 (1995); E. Weimar-Woods, \textit{The three-dimensional real Lie algebras and their contractions}, Jour. Math. Phys. \textbf{32} (1991) 2028; E. Weimar-Woods, \textit{Contractions, generalized In\"{o}n\"{u} and Wigner contractions and deformations of finite-dimensional Lie algebras}, Rev. Math. Phys. \textbf{12} 1505-1529 (2000)

\bibitem{Saletan} E.J. Saletan, \textit{Contractions of Lie groups}, J.Math.
Phys. \textbf{2}, 1-21 (1961).

\bibitem{def1} M. Gerstenhaber, \textit{On the deformations of rings and
algebras}, Ann. Math. \textbf{79}, 59-103 (1964)

\bibitem{def2} A. Nijenhuis and R.W. Richardson Jr., \textit{Cohomology and
deformations in graded Lie algebras}, Bull. A, Math. Soc. \textbf{72}, 1-29
(1966); A. Nijenhuis and R.W. Richardson Jr., \textit{Deformations of
Lie algebra structures}, J. Math. Mech. \textbf{171}, 89-105 (1967)

\bibitem{def4} R.W. Richardson, \textit{On the rigidity of semi-direct
products of Lie algebras}, Pac. J. Math. \textbf{22}, 339-344 (1967)

\bibitem{AzLibro} J.A. de Azc\'{a}rraga and J.M. Izquierdo, \textit{Lie
groups, Lie algebras, cohomology and some applications in physics}, Camb.
Univ. Press. (1995).

\bibitem {hs}M. Hatsuda and M. Sakaguchi, \textit{Wess-Zumino Term for the AdS
Superstring and Generalized In\"{o}n\"{u}-Wigner Contraction} Prog. Theor.
Phys. \textbf{109}, 853 (2003), e-print arXiv:hep-th/0106114.

\bibitem {aipv1}J. A. de Azcarraga, J. M. Izquierdo, M. Picon, and O. Varela,
\textit{Generating Lie and Gauge Free Differential (Super) \'{A}lgebras by
Expanding Maurer-Cartan Forms and Chern-Simons Supergravity} Nucl. Phys. B
\textbf{662} (2003) , 185. arXiv: hep-th/0212347

\bibitem {aipv2}J.A. de Azcarraga, J.M. Izquierdo, M. Picon and O. Varela, \textit{Extensions, expansions, Lie algebra cohomology and enlarged superspaces},
Class. \& Quant. Grav \textbf{21} (2004) S 1375. $\left[  \text{arXiV:hep-th/0401033}%
\right]  .$

\bibitem {aipv3}J.A. de Azcarraga, J.M. Izquierdo, M. Picon and O. Varela, \textit{Expansions of algebras and superalgebras and some applications},
Int.J.Theor.Phys. \textbf{46} 2738-2752 (2007) $\left[  \text{arXiV:hep-th/0703017}%
\right]  $

\bibitem {irs}Fernando Izaurieta, E. Rodriguez, P. Salgado, \textit{Expanding
Lie (super)algebras through Abelian semigroups} Jour. of Math. Phys.
\textbf{47} (2006) 123512 [arXiv:heo-th/0606215].

\bibitem {Nesterenko1} Burde D., \textit{Degenerations of 7-dimensional nilpotent Lie algebras} Comm. Algebra \textbf{33} (2005), 1259-1277

\bibitem {Nesterenko6} Nesterenko M. and Popovych R., \textit{Contractions of low-dimensional Lie algebras}, J. Math. Phys. \textbf{47} (2006), 123515

\bibitem {Nesterenko7} Popovych D.R. and Popovych R.O., \textit{Lowest dimensional example on non-universality of generalized In\"{o}n\"{u}-Wigner contractions}, J. Algebra \textbf{324} (2010), 2742-2756.

\bibitem {Nesterenko} Nesterenko M., \textit{S-expansions of three dimensional Lie algebras}, arXiv:1212.1820v1 [math-ph] (2012).

\bibitem {cham} A. H. Chamseddine, \textit{Topological Gauge Theory of Gravity in Five and All Odd Dimensions}, Phys. Lett. B \textbf{233} (1989) 291.

\bibitem {zan}J. Zanelli, \textit{Lecture notes on Chern-Simons (super-)gravities}, arXiv:hep-th/0502193v4

\bibitem {olea}R. Aros, M. Contreras, R. Olea, R. Troncoso, J. Zanelli, \textit{Charges in (2 + 1)-dimensional Gravity and Supergravity} ’99 Conf. (Potsdam, Germany, 19–24 July)

\bibitem {olea2}P. Mora, R. Olea, R. Troncoso, J. Zanelli, \textit{Finite Action Principle for Chern–Simons AdS Gravity}. J. High Energy Phys. JHEP06(2004)036 
arXiv: hep-th/0405267; P. Mora, R. Olea, R. Troncoso, J. Zanelli, \textit{Transgression Forms and Extensions of Chern–Simons Gauge Theories}. J. High Energy Phys. JHEP02(2006)067 arXiv: hep-th/0601081.

\bibitem {olea3}P. Mora, R. Olea, R. Troncoso, J. Zanelli, \textit{Vacuum Energy in Odd-Dimensional AdS Gravity}. arXiv: hep-th/0412046.

\bibitem {olea4}R. Olea, Mass,\textit{ Angular Momentum and Thermodynamics in four-dimensional Kerr-AdS Black Holes}. J. High Energy Phys. JHEP06(2005)023 arXiv: hep-th/0504233.

\bibitem {mora}P. Mora, \textit{Transgression Forms as Unifying Principle in Field Theory}. Ph.D. Thesis, Universidad de la Republica, Uruguay (2003). arXiv: hep-th/0512255; P. Mora, \textit{Unified Approach to the Regularization of Odd Dimensional AdS Gravity}. arXiv: hep-th/0603095.

\bibitem {irs1}F. Izaurieta, E. Rodriguez, P. Salgado,
\textit{Eleven-Dimensional Gauge Theory for the M Algebra as an Abelian
Semigroup Expansion of }$\mathfrak{osp}$\textit{(32/1).} Eur. Phys. Jour. C \textbf{54} 675–84 (2008)

\bibitem {irs2}F. Izaurieta, A. P\'{e}rez, E. Rodr\'{\i}guez, P. Salgado,
\textit{Dual formulation of the Lie Algebra S-expansion Procedure}, Jour.
Math. Phys. \textbf{50} 7 073511 (2009).

\bibitem {edelstein} J. D. Edelstein, M. Hassa\"{i}ne, R. Troncoso, J. Zanelli, \textit{Lie-algebra Expansions, Chern–Simons theories and the Einstein–Hilbert Lagrangian}. Phys. Lett. B \textbf{640} (2006) 278. arXiv: hep-th/0605174

\bibitem {K15}Fernando Izaurieta, Paul Minning, Alfredo P\'{e}rez, Eduardo
Rodr\'{\i}guez, Patricio Salgado, \textit{Standard General Relativity from
Chern-Simons Gravity}, Phys. Lett. B \textbf{678} (2009) 213-217.

\bibitem {quinz} C. A. C. Quinzacara and P. Salgado, \textit{Black hole for the Einstein-Chern-Simons gravity}, Physical Review D \textbf{85}, 124026 (2012)

\bibitem {gomez} F. Gomez, P. Minning, and P. Salgado, \textit{Standard cosmology in Chern-Simons gravity}, Physical Review D \textbf{84}, 063506 (2011)

\bibitem {K11}R. Caroca, N. Merino, P. Salgado, \textit{S-Expansion of
Higher-Order Lie Algebras}, Jour. Math. Phys. \textbf{50}, 013503 (2009).

\bibitem {K12}R. Caroca, N. Merino, A. P\'{e}rez, P. Salgado,
\textit{Generating higher order Lie algebras by expanding Maurer-Cartan
forms}, Jour. Math. Phys. \textbf{50}, 123527 (2009).

\bibitem {K10}J.A. de Azc\'{a}rraga, J.C. P\'{e}rez Bueno,
\textit{Higher-order simple Lie Algebras}, Commun. Math. Phys. \textbf{184} 669 (1997).

\bibitem {Spin3}Tom Lada and Jim Stasheff, \textit{Introduction to SH Lie
algebras for physicists}, (1992) [arXiv:hep-th/9209099v1].

\bibitem {Spin1}Anders K. H. Bengtsson, \textit{Towards Unifying Structures in
Higher Spin Gauge Symmetry}, (2008) [arXiv:hep-th/0802.0479v1].

\bibitem {Spin2}Anders K. H. Bengtsson, \textit{Structure of higher spin gauge
interactions}, Journal of Mathematical Physics \textbf{48}, 072302 (2007).

\bibitem {K14}R.~Caroca, N.~Merino, P.~Salgado and O.~Valdivia,
\textit{Generating infinite-dimensional algebras from loop algebras
by expanding Maurer-Cartan forms},  J. Math. Phys. \textbf{52} (2011) 043519.

\bibitem {V0}X. Bekaerta, S. Cnockaertb, C. Iazeollac and M.A.Vasiliev,
\textit{Nonlinear Higher Spin Theories in Various Dimensions}, (2005)
[arXiv:hep-th/0503128v2].

\bibitem {V1}M.A. Vasiliev, \textit{Higher spin gauge theories in various
dimensions}, Fortschr. Phys. \textbf{52}, No. 6-7, 702-717 (2004).

\bibitem {V2}M.A. Vasiliev, \textit{Actions, Charges and Off-Shell Fields in
the Unfolded Dynamics Approach}, (2005) [arXiv:hep-th/0504090v3].

\bibitem {V3}M.A. Vasiliev, \textit{Cubic Interactions of Bosonic Higher Spin
Gauge Fields in }$AdS_{5}$, (2004) [arXiv:hep-th/0106200v4].

\bibitem {V4}M.A. Vasiliev, \textit{Higher Spin Superalgebras in any Dimension
and their Representations}, (2005) [arXiv:hep-th/0404124v4].

\bibitem {Diaz} J. D\'{i}az, O. Fierro, F. Izaurieta, N. Merino, E. Rodriguez, P. Salgado and O. Valdivia, \textit{A generalized action for (2+1)-dimensional Chern–Simons gravity}, J. Phys. A: Math. Theor. \textbf{45} (2012) 255207 (14pp)

\bibitem {Bianchi}R. Caroca, I. Kondrashuk, N. Merino and F. Nadal,
\textit{Bianchi spaces and its }$3$\textit{-dimensional isometries as }%
$S$-expansion\textit{s of }$2$\textit{-dimensional isometries}, J. Phys. A: Math. Theor. \textbf{46} (2013) 225201 (24pp), arXiv: math-ph/1104.3541.

\bibitem {Nadal}N. Merino, F. Nadal, \textit{A Java library to perform S-expansions of Lie algebras} (in preparation)

\bibitem {n4}Forsythe G. E. 1955 \textit{SWAC computes 126 distinct semigroups of order 4}, Proc. Am. Math. Soc. \textbf{6} pp 443-447

\bibitem {n5}T.S. Motzkin and J.L. Selfridge 1956 \textit{Semigroups of
order five}, The November meeting in Los Angeles, Bull. Amer.
Math. Soc. \textbf{62} (1) 13-23.

\bibitem {n6-1}Plemmons R 1970 \textit{Construction and analysis of non-equivalent finite semigroups}, Computational Problems in
Abstract Algebra; Proceeding of Conference (Oxford, 1967) (Oxford: Pergamon) pp 223-228


\bibitem {n6-2}R. Plemmons 1969 \textit{A survey of computer applications to semigroups
and related structures}, ACM SIGSAM \textbf{12}, 28-39.

\bibitem {n6-3}Plemmons R 1967 \textit{There are 15973 semigroups of order 6}, Math. Algorithms \textbf{2} pp 2-17

\bibitem {n7}H. J\"{u}rgensen, P. Wick 1977 \textit{Die halbgruppen der ordnungen $\leq7$}, Semigroup Forum \textbf{14} , 69-79.

\bibitem {n8}Satoh S, Yama K and Tokizawa M 1994 \textit{Semigroups of order 8}, Semigroup Forum \textbf{49} pp 7-29

\bibitem {n9-1}Distler A and Kelsey T 2009 \textit{The monoids of orders eight, nine and ten} Ann. Math. Artif. Intell. \textbf{56} pp 3-21

\bibitem {n9-2} Distler A and Kelsey T 2008 \textit{The monoids of order eight and nine}, Intelligent Computer Mathematics: 9th
Int. Conf. on Artificial Intelligence and Symbolic Computation (Birmingham, July 2008) ed S Autexier,
J Campbell, J Rubio, V Sorge, M Suzuki and F Wiedijk (Lecture Notes in Computer Science vol 5144) 2008
(Berlin: Springer) pp 61-76


\bibitem {n9-3}A. Distler and J. D. Mitchell Smallsemi -- a GAP package, version 0.6.4, 2011. http://tinyurl.com/jdmitchell/smallsemi/


\bibitem {n9-4}A. Distler, T. Kelsey and J.D. Mitchell, \textit{Enumeration of Semigroups of Order 9} http://www-circa.mcs.stand.ac.uk/ShowcaseSlides/GatewayTalk.pdf

\bibitem {L2}A. O. Barut y R. Ratzka, \textit{Theory of group representations
and applications} (Singapore:World Scientific), 1986

\bibitem {L4}S. Helgason, \textit{Differential Geometry, Lie groups and
symmetric spaces}, (London: Academic) 1978.

\bibitem {Plemmons}Hildebrant J \textit{Handbook of Finite Semigroup Programs}, LSU Mathematics Electronic Preprint Series, preprint 2001-24; Plemmons R 1969 \textit{A survey of computer applications to semigroups and related structures} ACM SIGSAM Bulletin 12 pp 28-39

\bibitem{L3} R. Gilmore, \textit{Lie groups, Lie algebras, and some of their applications}, (New York: Wiley-Interscience) 1974.

\bibitem {Inos}C. Inostrosa, N. Merino, F. Nadal and P. Salgado \textit{Standard General Relativity from Chern-Simons Gravity, an alternative procedure}, (in preparation)

\end{thebibliography}
\end{document}